\documentclass[reqno]{amsart}
\usepackage{amssymb,graphics,graphicx,bbm,hyperref,color,longtable,ams-rust}

\usepackage{multicol}
\usepackage{enumitem}

\definecolor{gray}{rgb}{0.93,0.93,0.93}
\definecolor{light-gold}{rgb}{0.99,0.97,0.78}

\def\be{\begin{equation}}
\def\ee{\end{equation}}
\def\bm{\begin{multline}}
\def\bfig{\begin{figure}[htb]}
\def\efig{\end{figure}}
\newcommand{\dd}{{\rm d}}
\newcommand{\e}[1]{\,{\rm e}^{#1}\,}
\newcommand{\ii}{{\rm i}}
\def\Tr{{\operatorname{Tr\,}}}
\renewcommand{\Re}{{\rm Re}\;}

\newcommand{\transit}[4]{\; \substack{#1 \;\;\;\;\; #2 \\ =\!=\!=\!= \\ #3 \;\;\;\;\; #4} \;}

\numberwithin{equation}{section}
\newtheorem{theorem}{Theorem}[section]
\newtheorem{proposition}[theorem]{Proposition}
\newtheorem{lemma}[theorem]{Lemma}
\newtheorem{corollary}[theorem]{Corollary}
\newtheorem{conjecture}{Conjecture}

\newcommand{\eps}{{\varepsilon}}
\newcommand{\caC}{{\mathcal C}}
\newcommand{\caE}{{\mathcal E}}
\newcommand{\caH}{{\mathcal H}}
\newcommand{\caL}{{\mathcal L}}
\newcommand{\caP}{{\mathcal P}}
\newcommand{\caT}{{\mathcal T}}
\newcommand{\caV}{{\mathcal V}}
\newcommand{\caW}{{\mathcal W}}
\newcommand{\bbC}{{\mathbb C}}
\newcommand{\bbE}{{\mathbb E}}
\newcommand{\bbN}{{\mathbb N}}
\newcommand{\bbP}{{\mathbb P}}
\newcommand{\bbR}{{\mathbb R}}
\newcommand{\bbS}{{\mathbb S}}
\newcommand{\bbZ}{{\mathbb Z}}

\newcommand{\Id}{\mathrm{\texttt{Id}}}

\makeatletter
  \def\tagform@#1{\maketag@@@{\tiny{(#1)}\@@italiccorr}}
\makeatother

\renewcommand{\eqref}[1]{(\ref{#1})}

\begin{document}


\title{Random loop representations for quantum spin systems}

\author{Daniel Ueltschi}
\address{Department of Mathematics, University of Warwick,
Coventry, CV4 7AL, United Kingdom}
\email{daniel@ueltschi.org}

\subjclass{60K35, 82B10, 82B20, 82B26, 82B31}

\keywords{Random loop model, quantum Heisenberg models, Mermin-Wagner theorem, reflection positivity}

\begin{abstract}
We describe random loop models and their relations to a family of quantum spin systems on finite graphs. The family includes spin $\frac12$ Heisenberg models with possibly anisotropic spin interactions and certain spin 1 models with SU(2)-invariance. Quantum spin correlations are given by loop correlations. Decay of correlations is proved in 2D-like graphs, and occurrence of macroscopic loops is proved in the cubic lattice in dimensions 3 and higher. As a consequence, a magnetic long-range order is rigorously established for the spin 1 model, thus confirming the presence of a nematic phase.
\end{abstract}

\thanks{Work partially supported by EPSRC grant EP/G056390/1.}
\thanks{\copyright{} 2013 by the author. This paper may be reproduced, in its
entirety, for non-commercial purposes.}

\maketitle

\vspace{-5mm}
\centerline{\em Dedicated to Elliott Lieb on the occasion of his eightieth birthday}

{\small\tableofcontents}

\section{Introduction}
\label{sec intro}

We study models of random loops that are closely related to quantum lattice systems. There are many reasons to consider them. First, the probabilistic setting is mathematically elegant, it is fairly simple to depict, and the presence of spatial correlations makes it very interesting. Second, quantum spin systems play a major r\^ole in our understanding of the electronic properties of condensed matter. And third, the relations between these models are fascinating. They hold in spite of their mathematical differences, and several properties happen to be visible in one representation but not in the other.

The models of random loops described in this article find their origin in the work of T\'oth on the quantum Heisenberg ferromagnet \cite{Toth1}, and in the work of Aizenman and Nachtergaele on the Heisenberg antiferromagnet \cite{AN}. The present extension allows to describe the spin $\frac12$ quantum XY model and certain SU(2)-invariant spin 1 spin models. A major advantage of these representations is that correlations between quantum spins are given in terms of properties of the loops. 

Quantum Hamiltonians are also related to the generators of evolution of classical interacting particles. This is an active research area and some methods have been borrowed from quantum systems, such as the Bethe ansatz. We refer to \cite{GS, SS, GKRV, JK} for recent discussions. A Lorentz gas model has been recently introduced in \cite{Lef} that seems closely related to the representation of the Heisenberg ferromagnet. Further connections should be possible and the probabilistic representations of quantum lattice systems at equilibrium could help to bridge the gap.

Many important results have been obtained for quantum spin systems that are relevant for models of random loops. One such result is the Mermin-Wagner theorem about the absence of continuous symmetry breaking in dimensions 1 and 2 \cite{MW,FJ,MS,Pfi,FP1,FP2,KT,ISV,Nac}. We provide a simple proof; it uses operator theory and it is restricted to those parameters for which there exists a correspondence with quantum models. An interesting challenge is to find a probabilistic proof that holds more generally.

Another important result in statistical mechanics is the proof of existence of a phase transition with continuous symmetry breaking and long-range order in dimensions 3 and higher. The first positive result is due to Fr\"ohlich, Simon, and Spencer, for the classical Heisenberg model \cite{FSS}. The extension to quantum models was achieved by Dyson, Lieb, and Simon \cite{DLS}. They used the method of reflection positivity and infrared bounds, that has been subsequently discussed and extended in \cite{NP,FILS1,FILS2,KLS1,KLS2, ALSSY, AFFS}. It is worth emphasizing that continuous symmetry breaking has, so far, been proved only when reflection positivity is available, which excludes the Heisenberg ferromagnet.

The main goal of this article is to show that this method can also be applied to the models of random loops. As a result, we prove the occurrence of macroscopic loops when the spin parameter $S$ is small and the dimension of the lattice is large enough. In the case $S=\frac12$, we recover the results of \cite{DLS,KLS1} for the XY and Heisenberg models. The present approach becomes useful, from the perspective of quantum spin systems, in the case $S=1$ since it allows to prove the existence of quadrupolar long-range order at low temperature. In addition, the representation of random loops suggests that correlations of the form $\langle S_{x}^{3} S_{y}^{3} \rangle$ decay exponentially fast with respect to $\|x-y\|$. The results of this article are compatible with the presence of a spin nematic phase, that was studied in \cite{BO,TZX,TLMP,FKK}.

The model of random loops is introduced in Section \ref{sec rep}. We discuss the family of quantum systems in Section \ref{sec quantum} where we state and prove the connections to random loops. Correlations for SU(2)-invariant systems are particularly interesting, see Theorem \ref{thm integer spin}. The Mermin-Wagner theorem is discussed in Section \ref{sec MW}. We prove it using a simple inequality that, in essence, estimates the difference of expectations with respect to two Gibbs states in term of their relative entropy. The occurrence of macroscopic loops is described in Sections \ref{sec macroscopic loops} and \ref{sec sp rp}. We discuss explicit models with spin $\frac12$ and 1 in Section \ref{sec models}. Theorems \ref{thm spin 1} and \ref{thm spin 1 af} contain the new results for quantum systems obtained in this article. We conclude with a heuristic discussion about the joint distribution of the lengths of macroscopic loops, and of the nature of pure Gibbs states of the quantum models.

\medskip
{\it Acknowledgments:} It is a pleasure to thank Michael Aizenman, J\"urg Fr\"ohlich, Martin Hairer, Bruno Nachtergaele, Charles-\'Edouard Pfister, Robert Seiringer, and Tom Spencer, for valuable discussions. I am grateful to Marek Biskup and Roman Koteck\'y for their invitation to give lectures on related material in Prague in September 2011, and to Christian Hainzl and Stefan Teufel for their invitation to T\"ubingen in July 2012. I am also indebted to Maria Esteban and Mathieu Lewin for organizing the valuable program {\it Variational and spectral methods in quantum mechanics} at the Institut Henri Poincar\'e in 2013. Hong-Hao Tu kindly pointed out relevant references for the spin 1 model. The referee made several useful observations.

\section{Probabilistic models}
\label{sec rep}

We start by introducing the model of random loops in Section \ref{sec loops}, then the model of space-time spin configurations in Section \ref{sec stsc}. The former is more elegant and it contains all the relevant events for the description of quantum correlations. The latter contains more information and it will be used in Section \ref{sec macroscopic loops} that discusses the occurrence of macroscopic loops in rectangular boxes. The model of random loops is a marginal of the second model.

\subsection{Model of random loops}
\label{sec loops}

Let $(\Lambda,\caE)$ denote a finite graph, with $\Lambda$ the set of vertices, and $\caE$ the set of edges. An edge is always between distinct vertices, and two vertices are connected by at most one edge. The most relevant example is the box $\Lambda = \{1,\dots,L\}^{d}$ in $\bbZ^{d}$, with $\caE$ the set of nearest-neighbors, $\caE = \{ \{x,y\} \subset \Lambda : \|x-y\|_{1} = 1 \}$, but we consider more general graphs. Let $\beta>0$ be a parameter. We consider a Poisson point process on $\caE \times [0,\beta]$ with two distinct events, the ``crosses'' and the ``double bars''. Let $\caW_{k}$ denote the collection of subsets of $\caE \times [0,\beta]$ with cardinality $k$. The set $\Omega$ of the realizations of our Poisson point process is then
\be
\Omega = \bigcup_{k_{1}\geq0} \caW_{k_{1}} \times \bigcup_{k_{2}\geq0} \caW_{k_{2}}.
\ee
Let $u \in [0,1]$ be another parameter. We let $\rho_{u}$ denote the probability measure on $\Omega$ that corresponds to Poisson point processes of intensity $u$ for the crosses and $1-u$ for the double bars.

To a given realization $\omega \in \Omega$ of the Poisson point process corresponds a set of loops, denoted $\caL(\omega)$. The notion of loops is best understood by looking at pictures, see Fig.\ \ref{fig loops}. It is a necessary pain to go through mathematically precise definitions, though. So a {\it loop} of length $\ell$ is a closed trajectory $\gamma: [0,\beta\ell]_{\rm per} \to \Lambda \times [0,\beta]_{\rm per}$ such that
\begin{itemize}
\item $\gamma(s) \neq \gamma(s')$ if $s, s'$ are distinct points of continuity of $\gamma$.
\item $\gamma(s)$ is piecewise differentiable, and $\gamma'(s) = \pm1$ at points of differentiability.
\item If $s$ is a point of non-differentiability, then $\{ \gamma(s-), \gamma(s+)\} \in \caE \times [0,\beta]$.
\end{itemize}
We identify loops that have the same support but different parametrizations, i.e., the loops described by $\gamma(s+t)$ and $\gamma(-s)$ are considered to be identical to $\gamma(s)$.

\begin{centering}
\bfig
\begin{picture}(0,0)%
\includegraphics{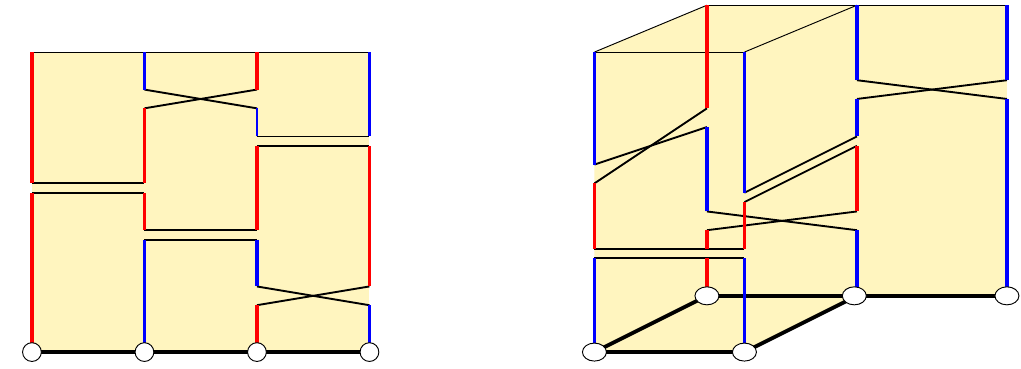}%
\end{picture}%
\setlength{\unitlength}{2368sp}%
\begingroup\makeatletter\ifx\SetFigFont\undefined%
\gdef\SetFigFont#1#2#3#4#5{%
  \reset@font\fontsize{#1}{#2pt}%
  \fontfamily{#3}\fontseries{#4}\fontshape{#5}%
  \selectfont}%
\fi\endgroup%
\begin{picture}(8159,2937)(1246,-4090)
\put(5731,-4036){\makebox(0,0)[lb]{\smash{{\SetFigFont{8}{9.6}{\rmdefault}{\mddefault}{\updefault}{\color[rgb]{0,0,0}$0$}%
}}}}
\put(7876,-3961){\makebox(0,0)[lb]{\smash{{\SetFigFont{8}{9.6}{\rmdefault}{\mddefault}{\updefault}{\color[rgb]{0,0,0}$\Lambda$}%
}}}}
\put(4351,-4036){\makebox(0,0)[lb]{\smash{{\SetFigFont{8}{9.6}{\rmdefault}{\mddefault}{\updefault}{\color[rgb]{0,0,0}$\Lambda$}%
}}}}
\put(1276,-1636){\makebox(0,0)[lb]{\smash{{\SetFigFont{8}{9.6}{\rmdefault}{\mddefault}{\updefault}{\color[rgb]{0,0,0}$\beta$}%
}}}}
\put(5776,-1636){\makebox(0,0)[lb]{\smash{{\SetFigFont{8}{9.6}{\rmdefault}{\mddefault}{\updefault}{\color[rgb]{0,0,0}$\beta$}%
}}}}
\put(1246,-4036){\makebox(0,0)[lb]{\smash{{\SetFigFont{8}{9.6}{\rmdefault}{\mddefault}{\updefault}{\color[rgb]{0,0,0}$0$}%
}}}}
\end{picture}%
\caption{Graphs and realizations of Poisson point processes, and their loops. In both cases, the number of loops is $|\caL(\omega)| = 2$.}
\label{fig loops}
\efig
\end{centering}

Let $\omega \in \Omega$, and $(x,t) \in \Lambda \times [0,\beta]$, but not in the support of $\omega$. The loop $\gamma$ that contains $(x,t)$ can be defined by starting with $\gamma(s) = (x,t+s)$ in a neighborhood of $s=0$, and by moving ``up'' until a cross or a double bar is met. If it is a cross, go across and continue in the same vertical direction. If it is a bar, go across and continue in the opposite vertical direction. The trajectory will eventually return to $(x,t)$ and the loop is complete. Let $\caL(\omega)$ be the set of all loops, starting from any point of $\Lambda \times [0,\beta]_{\rm per}$ (not in the support of $\omega$), modulo reparametrization. The number of loops is necessarily finite (it is always less than $|\Lambda| + |\omega|$).

Let $\theta>0$ be a parameter. We define the {\it partition function} as
\be
Y^{(u)}_{\theta}(\beta,\Lambda) = \int_{\Omega} \theta^{|\caL(\omega)|} \dd\rho_{u}(\omega).
\ee
The relevant probability measure for the model of random loops is then given by
\be
\label{La mais-je}
\frac1{Y^{(u)}_{\theta}(\beta,\Lambda)} \theta^{|\caL(\omega)|} \dd\rho_{u}(\omega).
\ee

Correlations will be given by three events, i.e., subsets of $\Omega$:
\begin{itemize}
\item $E_{x,y,t}^{+}$ is the set of all $\omega \in \Omega$ such that $(x,0)$ and $(y,t)$ belong to the same loop, and with identical vertical direction at these points:
\be
\label{pour E1}
\frac{\dd}{\dd s} \gamma(s) \Big|_{\gamma(s) = (x,0)} = \frac{\dd}{\dd s} \gamma(s) \Big|_{\gamma(s) = (y,t)}.
\ee
(This definition does not depend on the actual parametrization of $\gamma$.)
\item $E_{x,y,t}^{-}$ is the set of all $\omega \in \Omega$ such that $(x,0)$ and $(y,t)$ belong to the same loop, and with opposite vertical directions at these points:
\be
\label{pour E2}
\frac{\dd}{\dd s} \gamma(s) \Big|_{\gamma(s) = (x,0)} = -\frac{\dd}{\dd s} \gamma(s) \Big|_{\gamma(s) = (y,t)}.
\ee
\item $E_{x,y,t} = E_{x,y,t}^{+} \cup E_{x,y,t}^{-}$ is the set of all realizations $\omega \in \Omega$ such that $(x,0)$ and $(y,t)$ belong to the same loop.
\end{itemize}
These events are illustrated in Fig.\ \ref{fig correl}.
We also define the length $L_{x,t}$ of the loop that contains $(x,t) \in \Lambda \times [0,\beta]$ as the sum of all its vertical lines.

\bfig
\includegraphics[width=120mm]{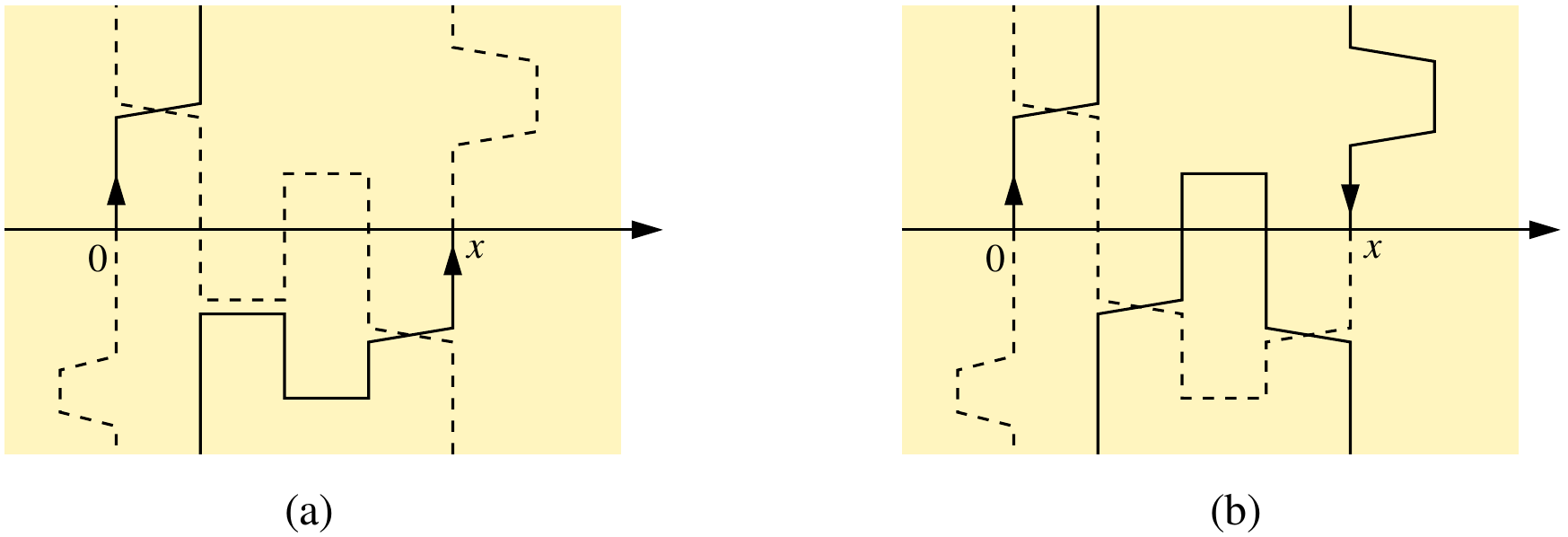}
\caption{Illustration for (a) the event $E_{0,x,0}^{+}$; (b) the event $E_{0,x,0}^{-}$.}
\label{fig correl}
\efig

The case $u=1$ and $\theta=1$ is the ``random interchange model'', or ``random stirring'' \cite{Har}, and loops are closely related to permutation cycles. The joint distribution of cycle lengths was obtained by Schramm for the complete graph \cite{Sch} following a conjecture of Aldous; see also \cite{Ber}. Infinite loops have been proved to occur on trees \cite{Ang, Ham1, Ham2}. Exact formul\ae{} for the probability of cyclic permutations have been obtained in \cite{AK,BK}. There is no correspondence between random loops with $\theta=1$ and quantum systems, and, rather curiously, none of the results obtained in the present article (decay of correlations in $d=2$ and occurrence of macroscopic loops in $d\geq3$) apply to the seemingly simpler situation $\theta=1$.

It is not hard to prove that, when $\beta$ is small, $\bbP(E_{x,y,t})$ decays exponentially fast with respect to the distance between $x$ and $y$. See e.g.\ Theorem 6.1 of \cite{GUW}.

\subsection{Space-time spin configurations}
\label{sec stsc}

Representations with space-time spin configurations have proved useful in many instances, see e.g.\ \cite{Gin, Ken, BKU, DFF, Gri, Iof, Gri,CI} and references therein. Given $S \in \frac12 \bbN$, a {\it space-time spin configuration} is a function
\be
\sigma : \Lambda \times [0,\beta]_{\rm per} \longrightarrow \{-S, -S+1, \dots, S\}.
\ee
such that $\sigma_{x,t}$ is piecewise constant in $t$, for any $x$. See Fig.\ \ref{fig spinconfig} for an illustration. Let $\Sigma$ denote the set of such functions with finitely-many discontinuities. We further require that these functions be c\`adl\`ag so they can be equipped with the Skorokhod topology. Non-experts should pay no attention to these technical details, they just help define a space of well-behaved functions on $\Sigma$.

\begin{centering}
\bfig
\includegraphics[width=90mm]{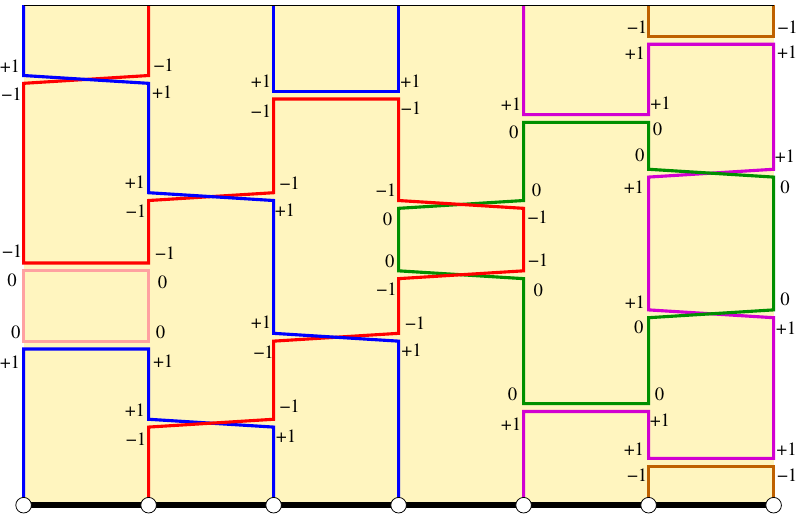}
\caption{(Color online) Illustration for a realization of the process $\rho_{\iota}$ and a compatible space-time spin configuration.}
\label{fig spinconfig}
\efig
\end{centering}

If $\omega$ is a realization of the Poisson process $\rho_{u}$ described above, we say that $\sigma$ is {\em compatible} with $\omega$ if $\sigma_{xt}$ is constant on each loop of $\caL(\omega)$. Let $\Sigma(\omega)$ denote the set of all such space-time spin configurations.
Notice that
\be
|\Sigma(\omega)| = (2S+1)^{|\caL(\omega)|}.
\ee
For $\theta=2S+1$ we have the relations
\be
Y^{(u)}_{\theta}(\beta,\Lambda) = \int \dd\rho_{u}(\omega) \sum_{\sigma \in \Sigma(\omega)} 1
\ee
and
\be
\label{ca sert plus tard}
\tfrac13 S(S+1) \bbP(E_{x,y,t}) = \frac1{Y^{(u)}_{\theta}(\beta,\Lambda)} \int\dd\rho_{u}(\omega) \sum_{\sigma \in \Sigma(\omega)} \sigma_{x,0} \sigma_{y,t}.
\ee
The latter relation follows from the following identity, which holds for all $S \in \frac12 \bbN$,
\be
\label{somme des carres}
\frac1{2S+1} \sum_{a=-S}^{S} a^{2} = \tfrac13 S(S+1).
\ee

We introduce now a more general setting for measures of space-time spin configurations. This will help us to make the connection between quantum spin systems and random loops. More importantly, we shall need the flexibility of this setting in Section \ref{sec macroscopic loops} in order to prove the existence of macroscopic loops.

It is enough for our purpose to consider only discontinuities that involve nearest-neighbor vertices. We introduce a Poisson point process on $\caE \times [0,\beta]$ whose objects are ``specifications'', that specify restrictions on the local configurations. Namely, let $\{x,y\} \in \caE$ and $t \in [0,\beta]$. An object is a set $A$ of allowed configurations at $x$ and $y$, immediately before and after $t$. The configuration at this ``transition'' is conveniently described as follows
\[
\transit{\sigma_{x,t+}}{\sigma_{y,t+}}{\sigma_{x,t-}}{\sigma_{y,t-}}
\]
For this notation to make sense we need an order on $\Lambda$; in the configuration above, we assume that $x<y$. A set $A$ is a subset of $\{-S,\dots,S\}^{4}$ and it appears with intensity $\iota(A)$. We are now ready for precise definitions.

Let $\iota : \caP \bigl( \{-S,\dots,S\}^{4} \bigr) \to \bbR_{+}$ denote the corresponding intensities ($\caP$ denotes the power set) and $\rho_{\iota}$ the Poisson point process. Given a realization $\xi$ of $\rho_{\iota}$, we let $\Sigma(\xi) \subset \Sigma$ denote the set of compatible space-time spin configurations, that is, $\sigma \in \Sigma(\xi)$ if
\begin{itemize}
\item $\displaystyle \transit{\sigma_{x,t+}}{\sigma_{y,t+}}{\sigma_{x,t-}}{\sigma_{y,t-}} \in A$ whenever $\xi$ contains the set $A$ at $(x,y,t)$.
\item $\sigma_{x,t}$ is constant in $t$ otherwise.
\end{itemize}
The measure is then given by $\rho_{\iota}$ and the counting measure on compatible configurations. It is important to remark that different intensities $\iota$ and $\iota'$ can give the same measure. They can be characterized by their ``atomic composition'' $\alpha_{\iota} : \{-S, \dots, S\}^{4} \to \bbR_{+}$:
\be
\alpha_{\iota} \Bigl( \transit{a}{b}{c}{d} \Bigr) = \sum_{\displaystyle A \ni \transit{a}{b}{c}{d}} \iota(A).
\ee

\begin{lemma}
\label{lem equiv intensities}
We have
\[
\int\dd\rho_{\iota}(\xi) \sum_{\sigma \in \Sigma(\xi)} F(\sigma) = \int\dd\rho_{\iota'}(\xi) \sum_{\sigma \in \Sigma(\xi)} F(\sigma)
\]
for all $F \in \caC(\Sigma)$, if and only if $\alpha_{\iota} = \alpha_{\iota'}$.
\end{lemma}

It is not hard to prove this lemma, e.g., by discretizing the interval $[0,\beta]$ and by comparing the sums on both sides. An immediate consequence is that Poisson point processes can be composed as follows.

\begin{lemma}
\label{lem add intensities}
Let $\iota, \iota'$ be intensities $\caP \bigl( \{-S,\dots,S\}^{4} \bigr) \to \bbR_{+}$. Then
\[
\int\dd\rho_{\iota}(\xi) \int\dd\rho_{\iota'}(\xi') \sum_{\sigma \in \Sigma(\xi \cup \xi')} F(\sigma) = \int\dd\rho_{\iota+\iota'}(\xi) \sum_{\sigma \in \Sigma(\xi)} F(\sigma)
\]
for any $F \in \caC(\Sigma)$.
\end{lemma}

Indeed, this follows from $\alpha_{\iota} + \alpha_{\iota'} = \alpha_{\iota+\iota'}$. Next, we relate the present setting with the random loops of the previous subsection.

\begin{lemma}
\label{lem on y est presque}
Assume that $\theta=2S+1$, and let the intensity $\iota$ be defined as follows:
\[
\iota \Bigl( \Bigl\{ \transit{a}{b}{b}{a} \Bigr\}_{a,b \in \{-S,\dots,S\}} \Bigr) = u, \qquad \iota \Bigl( \Bigl\{ \transit{b}{b}{a}{a} \Bigr\}_{a,b \in \{-S,\dots,S\}} \Bigr) = 1-u.
\]
We set $\iota(A)=0$ otherwise. Then
\[
 Y^{(u)}_{\theta}(\beta,\Lambda) = \int\dd\rho_{\iota}(\xi) \sum_{\sigma \in \Sigma(\xi)} 1
\]
and
\[
\tfrac13 S(S+1) \bbP(E_{x,y,t}) = \frac1{Y_{\theta}^{(u)}(\beta,\Lambda)} \int\dd\rho_{\iota}(\xi) \sum_{\sigma \in \Sigma(\xi)} \sigma_{x,0} \sigma_{y,t}.
\]
\end{lemma}

This follows from Eq.\ \eqref{ca sert plus tard}, once we observe that the sets above describe precisely the specifications given by the crosses and the double bars.

\section{Quantum spin systems}
\label{sec quantum}

\subsection{Families of quantum spin systems}

Let $S \in \frac12 \bbN$. The Hilbert space that describes the states of the system is the tensor product
\be
\caH_{\Lambda} = \bigotimes_{x\in\Lambda} \caH_{x},
\ee
where each $\caH_{x}$ is a copy of $\bbC^{2S+1}$. Let $S^{1}, S^{2}, S^{3}$ denote the usual spin operators on $\bbC^{2S+1}$. That is, they are hermitian matrices that satisfy
\begin{align}
\label{Pauli commutations}
&[S^{1}, S^{2}] = \ii S^{3}, \quad [S^{2}, S^{3}] = \ii S^{1}, \quad [S^{3}, S^{1}] = \ii S^{2},\\
&(S^{1})^{2} + (S^{2})^{2} + (S^{3})^{2} = S(S+1) \Id.
\label{square spins}
\end{align}
Recall that each $S^{i}$ has spectrum $\{-S,-S+1,\dots,S\}$. We use the notation $\vec S = (S^{1},S^{2},S^{3})$ and, for $\vec a \in \bbR^{3}$,
\be
S^{\vec a} = \vec a \cdot \vec S = a_{1} S^{1} + a_{2} S^{2} + a_{3} S^{3}.
\ee
These operators are related to the rotations in $\bbR^{3}$ by
\be
\e{-\ii S^{\vec a}} S^{\vec b} \e{\ii S^{\vec a}} = S^{R_{\vec a} \vec b}
\ee
where $R_{\vec a} \vec b$ denotes the vector $\vec b$ rotated around $\vec a$ by the angle $\|\vec a\|$.

Let $S_{x}^{i} = S^{i} \otimes \Id_{\Lambda\setminus\{x\}}$. We use Dirac's notation since it it very convenient. In $\bbC^{2S+1}$, $|a\rangle$ denotes the eigenvector of $S^{3}$ with eigenvalue $a \in \{-S,\dots,S\}$. In $\caH_{x} \otimes \caH_{y}$, $|a,b\rangle$ denotes the eigenvector of both $S^{3}_{x}$ and $S_{y}^{3}$ with respective eigenvalues $a$ and $b$.

We consider the three operators $T_{xy}$, $P_{xy}$, $Q_{xy}$ on $\caH_{\{x,y\}}$ (and their extensions on $\caH_{\Lambda}$ by identifying $T_{xy}$ with $T_{xy} \otimes \Id_{\Lambda\setminus\{x,y\}}$, etc...):
\begin{itemize}
\item $T_{xy}$ is the transposition operator:
\be
T_{xy} |a,b\rangle = |b,a\rangle.
\ee
\item $P_{xy}$ is the operator
\be
P_{xy} = \sum_{a,b=-S}^{S} (-1)^{a-b} |a, -a \rangle \langle b, -b|.
\ee
Equivalently, the matrix coefficients of $P_{xy}$ are given by
\be
\langle a,b| P_{xy} |c,d\rangle = (-1)^{a-c} \delta_{a,-b} \delta_{c,-d}.
\ee
Notice that $\frac1{2S+1} P_{xy}$ is the projector onto the spin singlet.
\item $Q_{xy}$ is identical to $P_{xy}$ except for the signs:
\be
\langle a,b| Q_{xy} |c,d\rangle = \delta_{a,b} \delta_{c,d}.
\ee
\end{itemize}
The first two operators are invariant under all rotations in $\bbR^{3}$, the last operator is invariant under rotations around the second direction of spins, as stated in the following lemma.

\begin{lemma}
\label{lem sym}
For all $\vec a \in \bbR^{3}$, we have
\begin{itemize}
\item[(a)] $\e{-\ii S_{x}^{\vec a} - \ii S_{y}^{\vec a}} T_{xy} \e{\ii S_{x}^{\vec a} + \ii S_{y}^{\vec a}} = T_{xy}$.
\item[(b)] $\e{-\ii S_{x}^{\vec a} - \ii S_{y}^{\vec a}} P_{xy} \e{\ii S_{x}^{\vec a} + \ii S_{y}^{\vec a}} = P_{xy}$.
\end{itemize}
And for $\vec a = a \vec e_{2}, a \in \bbR$; or $\vec a = a_{1} \vec e_{1} + a_{3} \vec e_{3}$ such that $a_{1}^{2} + a_{3}^{2} = \pi^{2}$,
\begin{itemize}
\item[(c)] $\e{-\ii S_{x}^{\vec a} - \ii S_{y}^{\vec a}} Q_{xy} \e{\ii S_{x}^{\vec a} + \ii S_{y}^{\vec a}} = Q_{xy}$.
\end{itemize}
\end{lemma}

\begin{proof}
It is not hard to check that, for any operator $A$ on $\bbC^{2S+1}$, we have
\be
\bigl[ A \otimes \Id + \Id \otimes A, T_{xy} \bigr] = 0.
\ee
Then $T_{xy}$ has SU($2S+1$) symmetry and (a) is a special case.

It is not straightforward to check (b) directly. But $\frac1{2S+1} P_{xy}$ is the projector onto the eigenspace of $(\vec S_{x} + \vec S_{y})^{2}$ for the eigenvalue 0. The eigenspace is known to have dimension 1, see e.g.\ \cite{Mes}, and the result follows. Finally, (c) follows from the relations
\be
Q_{xy} = \e{-\ii \pi S_{y}^{2}} P_{xy} \e{\ii \pi S_{y}^{2}} = \e{\ii \pi S_{y}^{2}} P_{xy} \e{-\ii \pi S_{y}^{2}},
\ee
and
\be
\e{-\ii \pi (S_{x}^{3} + S_{y}^{3})} \e{-\ii \pi S_{y}^{2}} \e{\ii \pi (S_{x}^{3} + S_{y}^{3})} = \e{\ii \pi S_{y}^{2}}.
\ee
\end{proof}

We consider two distinct families of Hamiltonians, indexed by the parameter $u \in [0,1]$:
\begin{align}
&H_{\Lambda}^{(u)} = -\sum_{\{x,y\} \in \caE} \Bigl( u T_{xy} + (1-u) Q_{xy} - 1 \Bigr), \label{def fam1} \\
&\tilde H_{\Lambda}^{(u)} = -\sum_{\{x,y\} \in \caE} \Bigl( u T_{xy} + (1-u) P_{xy} - 1 \Bigr). \label{def fam2}
\end{align}
The first family is convenient for the probabilistic representations, and it contains many relevant special cases for $S=\frac12$: The usual Heisenberg ferromagnet and antiferromagnet models, and the XY model. This is explained in Section \ref{sec spin12}. The second family is physically more relevant and the case $S=1$ is treated in details in Section \ref{sec spin1}.

Let $Z^{(u)}(\beta,\Lambda)$ and $\tilde Z^{(u)}(\beta,\Lambda)$ denote the corresponding partition functions:
\begin{align}
&Z^{(u)}(\beta,\Lambda) = \Tr_{\caH_{\Lambda}} \e{-\beta H_{\Lambda}^{(u)}}, \\
&\tilde Z^{(u)}(\beta,\Lambda) = \Tr_{\caH_{\Lambda}} \e{-\beta \tilde H_{\Lambda}^{(u)}}.
\end{align}
The expectation of the operator $A$ in the Gibbs state with Hamiltonian $H_{\Lambda}^{(u)}$ is
\be
\langle A \rangle = \frac1{Z^{(u)}(\beta,\Lambda)} \Tr A \e{-\beta H_{\Lambda}^{(u)}}.
\ee
We also consider the Schwinger functions; given $t \in [0,\beta]$ and two operators $A$ and $B$ on $\caH_{\Lambda}$, let
\be
\langle A; B \rangle(t) = \frac1{Z^{(u)}(\beta,\Lambda)} \Tr A \e{-(\beta-t) H_{\Lambda}^{(u)}} B \e{-t H_{\Lambda}^{(u)}}.
\ee
Notice that $\langle AB \rangle = \langle A; B \rangle(0)$.

\subsection{Random loop representations}

The representation of the Gibbs operator $\e{-\beta H}$ in terms of probabilistic objects, with $H$ a Schr\"odinger operator, goes back to Feynman's approach to the interacting Bose gas. Such representations of lattice systems has allowed many authors to prove the occurrence of phase transitions in anisotropic lattice systems \cite{Gin,Ken,BKU,DFF}. Conlon and Solovej used a random walk representation in order to obtain estimates on the free energy of the $S=\frac12$ Heisenberg ferromagnet \cite{CS}. Their result was improved by T\'oth using a loop representation \cite{Toth1}. A similar representation was introduced by Aizenman and Nachtergaele for the $S=\frac12$ Heisenberg antiferromagnet, and more generally for interactions of the form $P_{xy}$ \cite{AN}. It allows them to relate the quantum spin chain ($d=1$) to two-dimensional Potts and random cluster models. Nachtergaele has proposed extensions for higher spins in \cite{Nac1,Nac2}. With Bachmann, they recently used the representation for the classification of gapped ground states \cite{BN}.

In this section we show that the representations of T\'oth and Aizenman-Nachtergaele can be combined and extended to the families $H_{\Lambda}^{(u)}$ and $\tilde H_{\Lambda}^{(u)}$ defined in Eqs\ \eqref{def fam1} and \eqref{def fam2}. As it turns out, the representation holds for all $S\in\frac12\bbN$ in the case of the family $H_{\Lambda}^{(u)}$ but only for $S\in\bbN$ in the case of the family $\tilde H_{\Lambda}^{(u)}$.

The first identity between quantum system and loop model concerns the partition functions.

\begin{theorem}
For all $u\in[0,1]$, we have
\[
\int_{\Omega} (2S+1)^{|\caL(\omega)|} \dd\rho_{u}(\omega) = \begin{cases} Z^{(u)}(\beta,\Lambda) & \text{for all } S \in \frac12 \bbN, \\ \tilde Z^{(u)}(\beta,\Lambda) & \text{for all } S \in \bbN. \end{cases}
\]
\end{theorem}

\begin{proof}
Using Trotter's product formula,
\be
\label{ca se developpe}
\begin{split}
\Tr \e{-\beta H_{\Lambda}^{(u)}} &= \lim_{N\to\infty} \Tr \Bigl( \prod_{\{x,y\} \in \caE} \e{\frac\beta N (u T_{xy} + (1-u) Q_{xy} - 1)} \Bigr)^{N} \\
&= \lim_{N\to\infty} \Tr \Bigl( \prod_{\{x,y\}\in\caE} \Bigl[ 1 - \tfrac\beta N + \tfrac\beta N \bigl( u T_{xy} + (1-u) Q_{xy} \bigr) \Bigr] \Bigr)^{N} \\
&= \lim_{N\to\infty} \sum_{\sigma^{(1)}, \dots, \sigma^{(N)}} \prod_{i=1}^{N} \langle \sigma^{(i)}| \prod_{\{x,y\} \in \caE} \bigl[ 1 - \tfrac\beta N + \tfrac\beta N \bigl( u T_{xy} + (1-u) Q_{xy} \bigr) \bigr] |\sigma^{(i+1)}\rangle.
\end{split}
\ee
The sum is over $\sigma^{(i)} \in \{-S,\dots,S\}^{\Lambda}$ and we set $\sigma^{(N+1)} \equiv \sigma^{(1)}$. The transposition operator $T_{xy}$ yields the specification
\[
\Bigl\{ \transit{a}{b}{b}{a} \Bigr\}_{a,b \in \{-S,\dots,S\}}
\]
and the operator $Q_{xy}$ yields
\[
\Bigl\{ \transit{b}{b}{a}{a} \Bigr\}_{a,b \in \{-S,\dots,S\}}
\]
In the limit $N\to\infty$ we obtain the expression of Lemma \ref{lem on y est presque}. This proves the claim for $Z^{(u)}(\beta,\Lambda)$.

The proof for $\tilde Z^{(u)}(\beta,\Lambda)$ is similar, except for two differences:
\begin{itemize}
\item In order for the product of matrix elements to differ from 0, the spin in the loop must change sign when the loop changes its vertical direction (that is, at double bars).
\item The matrix elements of double bars are
\[
(-1)^{\sigma_{x_{j},t_{j}-} - \sigma_{x_{j},t_{j}+}}
\]
and they can be factorized with respect to the loops: $(-1)^{\sigma_{x_{j},t_{j}-}}$ for the loop coming from below, and $(-1)^{\sigma_{x_{j},t_{j}+}}$ for the loop coming from above. If the spin of the loop is even, the factors are all equal to 1. If the spin is odd, the factors are all equal to $-1$ and their product is 1 because there is an even number of them.
\end{itemize}
\end{proof}

The situation with half-integer spins is very different. Some loops receive negative weights, such as a loop with two neigboring sites and two transitions, one cross and one double bar (see Fig.\ \ref{fig bad loop}). The representation therefore involves signed measures (still real) and we lose the probabilistic nature. It is not clear whether such representations can be useful.

\begin{centering}
\bfig
\includegraphics{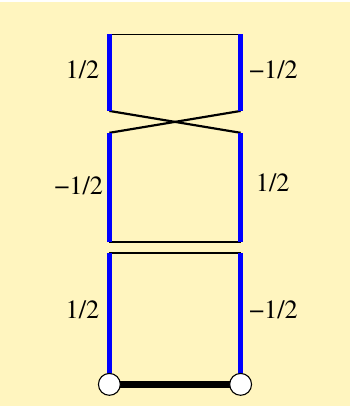}%
\caption{A ``bad loop'' on two vertices for the representation of the family $\tilde H_{\Lambda}^{(u)}$ with half-integer spin. The factor due to the double bar is $(-1)^{\frac12 - (-\frac12)} = -1$.}
\label{fig bad loop}
\efig
\end{centering}

We turn now to correlation functions. It is remarkable that the spin-spin correlations of the quantum models can be expressed in terms of properties of loops. Recall the events $E_{x,y,t}, E_{x,y,t}^{\pm}$ introduced above (with the help of Eqs \eqref{pour E1} and \eqref{pour E2}). We state first the results for the family $H_{\Lambda}^{(u)}$; the results for $\tilde H_{\Lambda}^{(u)}$ are postponed until Theorem \ref{thm integer spin}.

\begin{theorem}
\label{thm corr}
Consider the Hamiltonian $H_{\Lambda}^{(u)}$ with $S \in \frac12 \bbN$ and $u \in [0,1]$. Correlations in the spin directions 1 and 3 are given by
\[
\langle S_{x}^{1} ; S_{y}^{1} \rangle(t) = \langle S_{x}^{3} ; S_{y}^{3} \rangle(t) = \tfrac13 S(S+1) \bbP(E_{x,y,t}).
\]
Correlations in the spin direction 2 are given by
\[
\langle S_{x}^{2} ; S_{y}^{2} \rangle(t) = \tfrac13 S(S+1) \bigl[ \bbP(E_{x,y,t}^{+}) - \bbP(E_{x,y,t}^{-}) \bigr].
\]
\end{theorem}

\begin{proof}
Equality of correlations in directions 1 and 3 is clear by symmetry. We have
\be
\begin{split}
\langle S_{x}^{3} ; S_{y}^{3} \rangle(t) &= \frac1{Z^{(u)}(\beta,\Lambda)} \int\dd\rho_{u}(\omega) \sum_{\sigma\in\Sigma(\omega)} \sigma_{x,0} \sigma_{y,t} \\
&= \frac1{Z^{(u)}(\beta,\Lambda)} \int_{E_{x,y,t}} \dd\rho_{u}(\omega) (2S+1)^{|\caL(\omega)|} \Bigl( \frac1{2S+1} \sum_{a=-S}^{S} a^{2} \Bigr).
\end{split}
\ee
The result follows from the identity Eq.\ \eqref{somme des carres}.

For the correlations in the direction 2, we write a similar expansion as in Eq.\ \eqref{ca se developpe} but with additional factors $\langle \sigma_{x,0-} | S_{x}^{2} | \sigma_{x,0+} \rangle$ and $\langle \sigma_{y,t-} | S_{y}^{2} | \sigma_{y,t+} \rangle$. These factors force $(x,0)$ and $(y,t)$ to be in the same loop.  Now recall that $S^{2}= \frac1{2\ii} (S^{+}-S^{-})$ while $S^{1} = \frac12 (S^{+} + S^{-})$. If $\omega \in E_{x,y,t}^{+}$, there is one factor with $S^{+}$ and one factor with $S^{-}$, resulting in $-\ii^{2}$ times the same contribution as for $S^{1}$. If $\omega \in E_{x,y,t}^{-}$, on the other hand, both factors involve $S^{+}$ or both involve $S^{-}$, and the contribution is $\ii^{2}$ times that of $S^{1}$.
\end{proof}

Macroscopic loops are related to two physical properties of the system, namely spontaneous magnetization and magnetic susceptibility. This is stated in the following theorem.

\begin{theorem}\hfill
\label{thm a maintenant un label}
\begin{itemize}
\item[(a)] Macroscopic loops and magnetic susceptibility:
\[
\bbE\Bigl( \sum_{x\in\Lambda} L_{(x,0)} \Bigr) = \frac3{\beta S(S+1)} \frac{\partial^{2}}{\partial h^{2}} \log \Tr \e{-\beta H_{\Lambda}^{(u)} + \beta h \sum_{x\in\Lambda} S_{x}^{3}} \bigg|_{h=0}.
\]
\item[(b)] Macroscopic loops and the expectation of the square of the magnetization: There exists a constant $K$ (it depends on $S$ but not on $\beta,u,\Lambda,\caE$) such that
\[
\begin{split}
\frac{3\beta}{S(S+1)} \sum_{x,y\in\Lambda} \langle S_{x}^{3} S_{y}^{3} \rangle &- K \beta \sqrt{(1-u) |\caE|} \sqrt{\bbE \Bigl( \sum_{x\in\Lambda} L_{(x,0)} \Bigr)} \\
&\leq \bbE \Bigl( \sum_{x\in\Lambda} L_{(x,0)} \Bigr) \leq \frac{3\beta}{S(S+1)} \sum_{x,y\in\Lambda} \langle S_{x}^{3} S_{y}^{3} \rangle.
\end{split}
\]
\end{itemize}
\end{theorem}

The claim (a) implies that, for regular domains such as cubes, the infinite volume free energy is non analytic when the average length of the loops diverges. In the claim (b), the difference between the upper and lower bounds is smaller than the other terms when the loops have infinite average length.

\begin{proof}
Let us introduce the Duhamel two-point function:
\be
\label{def Duhamel}
(A,B)_{\rm Duh} = \frac1{Z^{(u)}(\beta,\Lambda)} \int_{0}^{\beta} \dd s \, \Tr A^{*} \e{-s H_{\Lambda}^{(u)}} B \e{-(\beta-s) H_{\Lambda}^{(u)}}.
\ee
Using Duhamel formula, we have
\be
\Tr \e{-\beta H_{\Lambda}^{(u)} + \beta h \sum_{x\in\Lambda} S_{x}^{3}} = Z^{(u)}(\beta,\Lambda) \Bigl[ 1 + \tfrac12 \beta h^{2} \sum_{x,y\in\Lambda} (S_{x}^{3}, S_{y}^{3})_{\rm Duh} + O(h^{4}) \Bigr].
\ee
Notice that the odd powers of $h$ do not contribute because of symmetry (rotation around $S^{2}$ of angle $\pi$). We have
\be
\begin{split}
\bbE\Bigl( \sum_{x\in\Lambda} L_{(x,0)} \Bigr) &= \sum_{x,y\in\Lambda} \int_{0}^{\beta} \bbP(E_{x,y,t}) \dd t \\
&= \frac3{S(S+1)} \sum_{x,y\in\Lambda} \int_{0}^{\beta} \langle S_{x}^{3} ; S_{y}^{3} \rangle(t) \dd t \\
&= \frac{3}{S(S+1)} \sum_{x,y\in\Lambda} (S_{x}^{3} , S_{y}^{3})_{\rm Duh} \\
&= \frac{3}{\beta S(S+1)} \frac{\partial^{2}}{\partial h^{2}} \log \Tr \e{-\beta H_{\Lambda}^{(u)} + \beta h \sum_{x\in\Lambda} S_{x}^{3}} \Big|_{h=0}.
\end{split}
\ee
This proves (a).

Let us recall the following inequalities that relate expectations with respect to Gibbs states and the Duhamel two-point function:
\be
\label{FB}
\begin{split}
\tfrac1\beta (A,A)_{\rm Duh} &\leq \tfrac12 \langle A^{*} A + A A^{*} \rangle \\
&\leq \tfrac12 \sqrt{(A,A)_{\rm Duh}} \sqrt{\langle [A^{*}, [H_{\Lambda}^{(u)},A]]\rangle} + \tfrac1\beta (A,A)_{\rm Duh}.
\end{split}
\ee
The first inequality follows from the convexity of the function $F(s) = \Tr A^{*} \e{-sH_{\Lambda}^{(u)}} A \e{-(\beta-s) H_{\Lambda}^{(u)}}$. The second inequality is more involved and is a consequence of the Falk-Bruch inequality, that was proposed independently in \cite{FB} and \cite{DLS}. Notice that the expectation of the double commutator is always nonnegative because it is equal to $([H_{\Lambda}^{(u)}, A],[H_{\Lambda}^{(u)}, A])_{\rm Duh}$. Using the inequality above with $A = \sum_{x} S_{x}^{3}$, we immediately get the upper bound in (b). We need to deal with the double commutator for the lower bound. It can be calculated, and we will actually need it in Section \ref{sec macroscopic loops}. Here, it is enough to notice that
\be
\bigr[ S_{x}^{3} + S_{y}^{3}, [T_{xy}, S_{x}^{3} + S_{y}^{3}] \bigl] = 0,
\ee
so that
\be
\Bigl\| \sum_{x,y \in \Lambda} [S_{x}^{3}, [H_{\Lambda}^{(u)}, S_{y}^{3}]] \Bigr\| \leq \text{const} (1-u) |\caE|.
\ee
Inserting in the second inequality in \eqref{FB}, we have
\be
\sum_{x,y\in\Lambda} \langle S_{x}^{3} S_{y}^{3} \rangle \leq {\rm const} \sqrt{\bbE\Bigl( \sum_{x\in\Lambda} L_{(x,0)} \Bigr)} \sqrt{(1-u) |\caE|} + \frac{S(S+1)}{3\beta} \bbE\Bigl( \sum_{x\in\Lambda} L_{(x,0)} \Bigr) .
\ee
The lower bound of (b) follows.
\end{proof}

Let us turn to the family of rotation invariant Hamiltonians $\tilde H^{(u)}_{\Lambda}$.

\begin{theorem}
\label{thm integer spin}
For the family $\tilde H_{\Lambda}^{(u)}$ with $S \in \bbN$ and $u\in[0,1]$, we have the following relations between spin and loop correlations:
\begin{itemize}
\item[(a)] $\langle S_{x}^{i} ; S_{y}^{i} \rangle(t) = \frac13 S(S+1) \bigl[ \bbP(E_{x,y,t}^{+}) - \bbP(E_{x,y,t}^{-}) \bigr]$.
\item[(b)] $\langle (S_{x}^{i})^{2} ; (S_{y}^{i})^{2} \rangle(t) - \langle (S_{x}^{i})^{2} \rangle \langle (S_{y}^{i})^{2} \rangle = \frac1{45} S(S+1)(2S-1)(2S+3) \bbP(E_{x,y,t})$.
\end{itemize}
\end{theorem}

The claim (a) suggests that correlations such as $\langle S_{x}^{i} S_{y}^{i} \rangle$ decay fast for $u \in (0,1)$. This will be a striking feature of the model with spin $S=1$ which is discussed in Section \ref{sec spin1}. Notice that if $u=1$, $E_{x,y,t}^{-}$ has probability zero. And if $u=0$ and the graph is bipartite, it is not too hard to check that $E_{x,y,t}^{+}$, resp.\ $E_{x,y,t}^{-}$, has probability zero if $x$ and $y$ belong to different sublattices, resp.\ identical sublattices.

\begin{proof}
The model is rotation invariant, so it is enough to prove it for $i=3$.  The correlation function of (a) is non-zero only if a loop connects $(x,0)$ and $(y,t)$. Spin values are identical if the vertical direction in the loop is identical, they are opposite if the direction in the loop is opposite. Hence the difference of probabilities of $E_{x,y,t}^{+}$ and $E_{x,y,t}^{-}$.

For (b), let us first observe that
\be
\langle (S_{x}^{i})^{2} \rangle = \tfrac13 \big\langle (S_{x}^{1})^{2} + (S_{x}^{2})^{2} + (S_{x}^{3})^{2} \big\rangle = \tfrac13 S(S+1).
\ee
Next,
\be
\begin{split}
\langle (S_{x}^{i})^{2} ; (S_{y}^{i})^{2} \rangle(t) &= \frac1{Z^{(u)}(\beta,\Lambda)} \int\dd\rho_{u}(\omega) \sum_{\sigma \in \Sigma(\omega)} \sigma_{x,0}^{2} \sigma_{y,t}^{2} \\
&= \frac1{Z^{(u)}(\beta,\Lambda)} \int_{E_{x,y,t}^{\rm c}} \dd\rho_{u}(\omega) (2S+1)^{|\caL(\omega)|} \biggl( \frac1{2S+1} \sum_{a=-S}^{S} a^{2} \biggr)^{2} \\
&\qquad + \frac1{Z^{(u)}(\beta,\Lambda)} \int_{E_{x,y,t}} \dd\rho_{u}(\omega) (2S+1)^{|\caL(\omega)|} \frac1{2S+1} \sum_{a=-S}^{S} a^{4} \\
&= \bigl( \tfrac13 S(S+1) \bigr)^{2} + \biggl[ \frac1{2S+1} \sum_{a=-S}^{S} a^{4} - \bigl( \tfrac13 S(S+1) \bigr)^{2} \biggr] \bbP(E_{x,y,t}).
\end{split}
\ee
We used the identity \eqref{somme des carres}. The claim follows from the identity
\be
\frac1{2S+1} \sum_{a=-S}^{S} a^{4} - \bigl( \tfrac13 S(S+1) \bigr)^{2} = \tfrac1{45} S(S+1)(2S-1)(2S+3),
\ee
which is valid for all integer and half-integer $S$.
\end{proof}

\section{Decay of correlations in 2D-like graphs}
\label{sec MW}

It is well-known that continuous symmetries cannot be broken in two spatial dimensions. This was first proved by Mermin and Wagner for the quantum Heisenberg model \cite{MW}. The original statement is about the absence of spontaneous magnetization in the lattice $\bbZ^{2}$. Many improvements have been made over the years, in particular the proof by Fr\"ohlich and Pfister that all KMS states are rotation invariant \cite{Pfi,FP1,FP2}. Here we are especially interested in the decay of the two-point correlation function. The first result for the Heisenberg model in $\bbZ^{2}$ is due to Fisher and Jasnow \cite{FJ}. Algebraic decay has been proved by McBryan and Spencer for the classical Heisenberg model, in a short and lucid article that introduces the idea of ``complex rotations'' \cite{MS}. There exists a beautiful extension to quantum systems by Koma and Tasaki \cite{KT}.

The following result is weaker than the one in \cite{KT}, which we did not know back then. The proof may still have its own interest as it is a bit simpler. It is somewhat inspired by \cite{Pfi,FP1,ISV,Nac}. A consequence is the absence of macroscopic loops in 2D-like graphs for $\theta = 2,3,4, \dots$. As the proof relies on the continuous symmetries that are present in the quantum setting, it does not seem possible to extend it to other values of $\theta$.

Let $d(x,y)$ denote the graph distance, i.e., the length of the minimal path that connects $x$ and $y$.

\begin{theorem}
\label{thm MW}
Assume that the graph $(\Lambda,\caE)$ is 2D-like, i.e., there exists a constant $C$ such that for all $x \in \Lambda$ and all integers $k$,
\[
\#\{ y \in \Lambda : d(x,y) = k \} \leq Ck.
\]
Then for all $u \in [0,1]$, $\beta \in [0,\infty)$, $S \in \frac12 \bbN$,
there exists a constant $K$ that depends on the graph only through $C$, independent of $x,y$, such that
\[
\langle S_{x}^{3} S_{y}^{3} \rangle = \tfrac13 S(S+1) \bbP(E_{x,y,0}) \leq \frac{K}{\sqrt{\log(d(x,y)+1)}}.
\]
\end{theorem}

We prove Theorem \ref{thm MW} with the help of the following inequality, that allows to compare expectations with respect to two Gibbs states, and something that is almost the relative entropy.

\begin{lemma}
\label{lem inegalite}
Let $A, H, H'$ be hermitian matrices such that $\Tr \e{-H} = \Tr \e{-H'} = 1$. Then for any $s>0$, we have
\[
\Tr A \e{-H} - \Tr A \e{-H'} \leq \frac1s \Tr(H'-H) \e{-H + sA} + s \|A\|^{2} \e{s\|A\|}.
\]
\end{lemma}

\begin{proof}
We start with the Taylor series with remainder; there exists $\eta(s) \in [0,s]$ such that
\be
\label{Taylor}
\Tr \e{-H + sA} = \Tr \e{-H} + s \Tr A \e{-H} + \tfrac12 s^{2} F(\eta(s)),
\ee
where $F(\eta)$ is the Duhamel two-point function
\be
\label{Duh}
F(\eta) = \int_{0}^{1} \Tr A \e{-t(H-\eta A)} A \e{-(1-t) (H-\eta A)} \dd t.
\ee
Let $F'$ be the same function but with $H'$ instead of $H$. We get
\be
\label{voila un label !}
\Tr A \e{-H} - \Tr A \e{-H'} = \frac1s \Tr \bigl( \e{-H+sA} - \e{-H'+sA} \bigr) - \frac s2 \bigl( F(\eta(s)) - F'(\eta'(s)) \bigr).
\ee
The remainder can be estimated using the convexity of the integrand in \eqref{Duh} (as function of $t$) and the minimax principle:
\be
0 \leq F(\eta) \leq \Tr A^{2} \e{-H+\eta A} \leq \|A\|^{2} \Tr \e{-H+\eta A} \leq \|A\|^{2} \e{\eta \|A\|}.
\ee
Applying Klein inequality to the middle term of \eqref{voila un label !}, we get the lemma.
\end{proof}

\begin{proof}[Proof of Theorem \ref{thm MW}]
We work in the quantum setting. Let $\phi_{z}$ be angles such that $\phi_{x}=\pi$ and $\phi_{y}=0$. We consider the unitary operator
\be
U = \prod_{z\in\Lambda} \e{\ii \phi_{z} S_{z}^{2}}.
\ee
Since $U^{*} S_{x}^{3} U = -S_{x}^{3}$ and $U^{*} S_{y}^{3} U = S_{y}^{3}$, we have
\be
\langle S_{x}^{3} S_{y}^{3} \rangle_{H_{\Lambda}^{(u)}} = -\langle S_{x}^{3} S_{y}^{3} \rangle_{U^{*} H_{\Lambda}^{(u)} U}.
\ee
Then
\be
\label{ca commence}
0 \leq \langle S_{x}^{3} S_{y}^{3} \rangle_{H_{\Lambda}^{(u)}} = \tfrac12 \Bigl[ \langle S_{x}^{3} S_{y}^{3} \rangle_{H_{\Lambda}^{(u)}} -\langle S_{x}^{3} S_{y}^{3} \rangle_{U^{*} H_{\Lambda}^{(u)} U} \Bigr].
\ee
We have
\be
\begin{split}
&\e{-\ii\phi_{z} S_{z}^{2} -\ii \phi_{z'} S_{z'}^{2}} \bigl( u T_{zz'} + (1-u) Q_{zz'} \bigr) \e{\ii\phi_{z} S_{z}^{2} +\ii \phi_{z'} S_{z'}^{2}} \\
&\quad= \e{-\ii (\phi_{z} - \phi_{z'}) S_{z}^{2}} \bigl( u T_{zz'} + (1-u) Q_{zz'} \bigr) \e{\ii (\phi_{z} - \phi_{z'}) S_{z}^{2}} \\
&\quad= u T_{zz'} + (1-u) Q_{zz'} + (\phi_{z}-\phi_{z'}) \bigl[ S_{z}^{2}, u T_{zz'} + (1-u) Q_{zz'} \bigr] + O\bigl( (\phi_{z}-\phi_{z'})^{2} \bigr).
\end{split}
\ee
The first equality follows from rotation invariance, see Lemma \ref{lem sym}.
Combining \eqref{ca commence} with Lemma \ref{lem inegalite} and the equation above, we get, for all $s \in (0,1]$,
\be
\begin{split}
\langle S_{x}^{3} S_{y}^{3} \rangle \leq &\frac1{s Z^{(u)}(\beta,\Lambda)} \sum_{\{z,z'\} \in \caE} (\phi_{z}-\phi_{z'}) \Tr \bigl[ S_{z'}^{2}, u T_{zz'} + (1-u) Q_{zz'} \bigr] \e{-\beta H_{\Lambda}^{(u)} + sS_{x}^{3} S_{y}^{3}} \\
&+ \frac{C_{1}}s \sum_{\{z,z'\} \in \caE} (\phi_{z} - \phi_{z'})^{2} + s C_{2}.
\end{split}
\ee
Here, $C_{1},C_{2}$ are constants that depend on $\beta$ and $S$, but they do not depend on the graph, nor on $s \in (0,1]$ and $x,y \in \Lambda$. The term with $C_{1}$ comes from a double commutator, and from higher order terms with multiple commutators.
The first term of the right side vanishes because of symmetry (rotation around $S^{3}$ by angle $\pi$). Had this symmetry not been available, a way out is to add another term with opposite angles, as in \cite{ISV,Nac}.

The following choice of $\phi_{z}$ is rather optimal:
\be
\phi_{z} = \begin{cases} \bigl( 1 - \frac{\log(d(x,z)+1)}{\log(d(x,y)+1)} \bigr) \pi & \text{if } d(x,z) \leq d(x,y), \\ 0 & \text{otherwise.} \end{cases}
\ee
It follows that $|\phi_{z}-\phi_{z'}| \leq \frac{C_{3}}{\log(d(x,y)+1)} \frac1{d(x,z)+1}$, so that
\be
\sum_{\{z,z'\}\in\caE} (\phi_{z}-\phi_{z'})^{2} \leq C_{3} \sum_{k=1}^{d(x,y)} Ck \frac1{(k \log(d(x,y)+1))^{2}} \leq \frac{C_{4}}{\log(d(x,y)+1}.
\ee
We eventually obtain, for any $s \in (0,1]$,
\be
0 \leq \langle S_{x}^{3} S_{y}^{3} \rangle \leq \frac{C_{5}}s \frac1{\log(d(x,y)+1} + s C_{2}.
\ee
We get the claim by choosing $s = 1/\sqrt{\log(d(x,y)+1)}$.
\end{proof}

\section{Occurrence of macroscopic loops in dimension $d\geq3$}
\label{sec macroscopic loops}

The occurrence of a phase transition at low temperature accompanied by spontaneous magnetization and symmetry breaking was first established by Fr\"ohlich, Simon, and Spencer for the classical Heisenberg model in dimension $d\geq3$ \cite{FSS}. They introduced the method of infrared bounds and reflection positivity. The difficult extension to quantum systems was done by Dyson, Lieb, and Simon for the Heisenberg model with antiferromagnetic interactions and large enough dimension (or spin) \cite{DLS}. The result was improved by Kennedy, Lieb, and Shastry to all $S\in\frac12\bbN$ and all $d\geq3$ \cite{KLS1}. Notice that these results cannot hold in $d\leq2$ as they would contradict Theorem \ref{thm MW} (Mermin-Wagner). (Symmetry breaking can take place in the ground state, i.e., in the limit of zero temperature \cite{NP,KLS2}.) All these results are proved using the method of reflection positivity and infrared bounds that was developed in \cite{FSS,DLS,FILS1,FILS2}. See \cite{ALSSY,AFFS} for recent advances. We recommend the Prague notes of T\'oth and Biskup for excellent introductions to the subject \cite{Toth2, Bis}, see also \cite{Uel}. The Vienna lectures of Fr\"ohlich offer an impressive glimpse of the background and of the context \cite{Fro}.

In this section we apply the method of infrared bounds and reflection positivity to the model of random loops. We introduce a model with external fields and cast its partition function in a reflection positive form. We do not use the methods of \cite{DLS, KLS1} directly, but these articles are lighting the way.

\subsection{Setting and results}

Let $\Lambda_{L}$ denote the cubic box in $\bbZ^{d}$ with side length $L$ and periodic boundary conditions, and let $(\Lambda_{L},\caE_{L})$ the graph where edges are nearest-neighbors. If a formal definition is needed, consider the quotient set $\Lambda_{L} = (\bbZ / L\bbZ)^{d}$. The Euclidean distance in $\bbZ^{d}$ has a natural extension in this set, and we define $\caE_{L}$ as the set of unordered pairs $\{x,y\} \subset \Lambda_{L}$ such that the distance between $x$ and $y$ is 1. In the sequel, we identify $\Lambda_{L}$ with the set
\be
\Lambda_{L} = \bigl\{ x \in \bbZ^{d} : -\tfrac L2 < x_{i} \leq \tfrac L2, i =1,\dots,d \bigr\}
\ee
Let $\eps(k)$ denote the ``dispersion relation'' of the discrete Laplacian,
\be
\eps(k) = 2 \sum_{i=1}^{d} (1-\cos k_{i}).
\ee

In order to state the main theorem, we need to introduce the following two integrals:
\be
\label{IJ}
\begin{split}
&I_{d} = \frac1{(2\pi)^{d}} \int_{[-\pi,\pi]^{d}} \sqrt{\frac{\eps(k+\pi)}{\eps(k)}} \Bigl( \frac1d \sum_{i=1}^{d} \cos k_{i} \Bigr)_{+} \dd k, \\
&J_{d} = \frac1{(2\pi)^{d}} \int_{[-\pi,\pi]^{d}} \sqrt{\frac{\eps(k+\pi)}{\eps(k)}} \dd k.
\end{split}
\ee
Here, $\eps(k+\pi) = 2\sum_{i=1}^{d} (1+\cos k_{i})$, and $(\cdot)_{+}$ denotes the positive part. One can check that, as $d\to\infty$, these integrals satisfy $I_{d} \to 0$ \cite{KLS2} and $J_{d}\to1$ \cite{DLS}.

\begin{theorem}
\label{thm irb}
Let $d\geq3$ and $u \in [0,\frac12]$. We have the two lower bounds
\[
\lim_{\beta\to\infty} \lim_{L\to\infty} \frac1{L^{d}} \sum_{x\in\Lambda} \bbP(E_{0,x,0}) \geq \begin{cases} 1 - \frac{2S+1}{\sqrt2} \, \sqrt{1-u} \; J_{d} \, \sqrt{\bbP(E_{0,e_{1},0})}; \\ \bbP(E_{0,e_{1},0}) - \frac{2S+1}{\sqrt2} \, \sqrt{1-u} \; I_{d} \, \sqrt{\bbP(E_{0,e_{1},0})}. \end{cases}
\]
\end{theorem}

In the right sides, $\bbP(E_{0,e_{1},0})$ denotes the probability that two neighbors belong to the same loop after taking the limits $L,\beta\to\infty$. Notice that all these limits exist but we do not prove it here. What we prove is the above statement with ``$\liminf$'' instead.
The proof of Theorem \ref{thm irb} can be found at the end of Section \ref{sec irb}.

Combining with Theorem \ref{thm a maintenant un label} (b), we obtain an estimate for $\bbE\bigl( \frac{L_{(0,0)}}{\beta L^{d}} \bigr)$, so that a positive lower bound implies the occurrence of macroscopic loops at low enough temperatures.
The claim also holds for $d=2$ but with the order of the limits over $\beta$ and $L$ interchanged.
The lower bounds involve $\bbP(E_{0,e_{1},0})$, which, curiously enough, needs to be small in the first bound and large in the second bound. We therefore get a positive lower bound whenever
\be
 \sqrt{\bbP(E_{0,e_{1},0})} > \tfrac{2S+1}{\sqrt2} \, \sqrt{1-u} \; I_{d} \qquad \text{or} \qquad  \sqrt{\bbP(E_{0,e_{1},0})} < \tfrac{\sqrt2}{(2S+1) \sqrt{1-u} \, J_{d}}.
 \ee
 At least one of these two inequalities holds when $\frac{2S+1}{\sqrt2} \sqrt{1-u} \, I_{d} < \frac{\sqrt2}{(2S+1) \sqrt{1-u} \, J_{d}}$, which is equivalent to
 \be
 \label{sufficient condition}
2S+1 < \bigl( \tfrac12 (1-u) I_{d} J_{d} \bigr)^{-\frac12}.
 \ee
Values of $I_{d}$ and $J_{d}$ can be found numerically. They are listed on Table \ref{table} for $2 \leq d \leq 6$. Occurrence of macroscopic loops (at temperatures low enough) follows for any $d\geq3$ provided that $S$ is small enough. Macroscopic loops are also present in the ground state in $d=2$ when $S=\frac12$ and $u$ is close to $\frac12$. The right side of \eqref{sufficient condition} diverges as $d\to\infty$ so that the result holds for arbitrary $S$ provided the dimension is large enough.
In the quantum Heisenberg model, the results get better when $S$ becomes large \cite{DLS}. It is natural that the model of random loops behaves differently. As $\theta = 2S+1$ increases, the number of loops also increases, and their size decreases. It would be interesting to know whether or not macroscopic loops occur for any $\theta>0$ at low enough temperatures, when $d\geq3$ is fixed.

The consequences of Theorem \ref{thm irb} in the cases $S=\frac12$ and $S=1$ are discussed in more details in Section \ref{sec models}.

\begin{table}[htb]
\centering
\colorbox{light-gold}{
\begin{tabular}{c|cccccc}
$d$ & $I_{d}$ & $J_{d}$ & $(\frac12 I_{d} J_{d})^{-\frac12}$ & \text{$u=0$} & $(\frac14 I_{d} J_{d})^{-\frac12}$ & \text{$u=\frac12$} \\
\hline
2 & 0.646803 & 1.39320 & 1.48978 & none & 2.10687 & $S=\frac12$ \\
3 & 0.349882 & 1.15672 & 2.22301 & $S=\frac12$ & 3.14381 & $S \leq 1$ \\
4 & 0.253950 & 1.09441 & 2.68256 & $S=\frac12$ & 3.79372 & $S \leq 1$ \\
5 & 0.206878 & 1.06754 & 3.00931 & $S \leq 1$ & 4.25581 & $S \leq \frac32$ \\
6 & 0.177716 & 1.05274 & 3.26958 & $S \leq 1$ & 4.62389 & $S \leq \frac32$ \\
\end{tabular}
}
\medskip
\caption{Numerical values of the integrals $I_{d}$ and $J_{d}$ defined in \eqref{IJ} and of the constant in the right side of Eq.\ \eqref{sufficient condition} for $d=2,\dots,6$. Values of $S$ at $u=0, \frac12$ for which the existence of macroscopic loops is proved in Theorem \ref{thm irb} at low temperatures (for $d=2$, this holds in the ground state only).}
\label{table}
\end{table}

\subsection{Reflection positivity of the model of random loops}

Let us introduce a notation for the loop correlation between $(0,0)$ and $(x,t)$:
\be
\kappa(x,t) = \bbP(E_{0,x,t}).
\ee
The Fourier transform plays an vital r\^ole and we use the following conventions:
\be
\label{kappa FT}
\begin{split}
&\widehat\kappa(k,t) = \sum_{x\in\Lambda_{L}} \e{-\ii kx} \kappa(x,t), \\
&\widetilde\kappa(k,\tau) = \sum_{x\in\Lambda_{L}} \int_{0}^{\beta} \dd t \e{-\ii kx - \ii\tau t} \kappa(x,t),
\end{split}
\ee
where $k$ belongs to the set
\be
\Lambda^{*} = \bigl\{ k \in \tfrac{2\pi}L \bbZ^{d} : -\pi < k_{i} \leq \pi, i=1,\dots,d \bigr\},
\ee
and $\tau \in \frac{2\pi}\beta \bbZ^{d}$. It is useful to write down the inverse transforms:
\be
\begin{split}
\kappa(x,t) &= \frac1{L^{d}} \sum_{k\in\Lambda^{*}} \e{\ii kx} \widehat\kappa(k,t), \\
&= \frac1{\beta L^{d}} \sum_{k\in\Lambda^{*}} \sum_{\tau \in \frac{2\pi}\beta \bbZ} \e{\ii kx + \ii\tau t} \widetilde\kappa(k,\tau).
\end{split}
\ee
Let us note that the Fourier transform of $\kappa$ is related to the expectation of the lengths of macroscopic loops:
\be
\bbE \Bigl( \frac{L_{(0,0)}}{\beta L^{d}} \Bigr) = \frac1{\beta L^{d}} \sum_{x\in\Lambda_{L}} \int_{0}^{\beta} \kappa(x,t) \dd t = \frac1{\beta L^{d}} \widehat\kappa(0,0).
\ee

The setting needs to be modified somewhat in order to reach a reflection positive form. Recall that $\sigma = (\sigma_{xt})$ denotes a space-time spin configuration where the possible values are the eigenvalues of spin operators: $\sigma_{xt} \in \{-S, -S+1, \dots, S\}$. We will need a more symmetric index set for $S\neq\frac12$, and we therefore consider the regular simplex $\caT_{2S+1}$ that contains $2S+1$ elements. We embed $\caT_{2S+1}$ in $\bbR^{2S}$ in such a way that the points are given by vectors $\vec a$, $a \in \{-S,\dots,S\}$, that satisfy
\be
\label{simplex}
\begin{split}
&\| \vec a - \vec b \| = 1 \text{ if }a \neq b; \\
&\vec a \cdot \vec b = -\frac1{2(2S+1)} \text{ if } a \neq b; \\
&\| \vec a \|^{2} = \frac S{2S+1}; \\
&\sum_{b=-S}^{S} \vec a \cdot \vec b = 0  \text{ for any fixed } a \in \{-S,\dots,S\}.
\end{split}
\ee
These properties are straightforward if we consider the projection of the unit vectors of $\bbR^{2S+1}$ onto the hyperplane perpendicular to $(1,1,\dots,1)$. More precisely, we define
\be
\vec a = \tfrac1{\sqrt2} \bigl( \vec e_{a} - \tfrac1{2S+1} (1,1,\dots,1) \bigr),
\ee
where $\vec e_{a}$ is the unit vector in $\bbR^{2S+1}$ that corresponds to $a$. Since all vectors $\vec a$ belong to a common $2S$-dimensional subspace, we can view them as elements of $\bbR^{2S}$.

We denote $\vec a^{(i)}$ the $i$th component of the vector $\vec a$.
Let $v$ be a real ``field'', i.e., a function $\Lambda_{L} \to \bbR$. The key object is the following partition function $Z(v)$:
\be
\label{fpart fields}
Z(v) = \int\dd\rho_{\iota}(\xi) \sum_{\sigma\in\Sigma(\xi)} \exp \biggl\{ -\sum_{\{x,y\} \in \caE} \int_{0}^{\beta} \dd t \Bigl[ (\vec\sigma_{xt}^{(1)} - \vec\sigma_{yt}^{(1)}) (v_{x}-v_{y}) + \tfrac14 (v_{x}-v_{y})^{2} \Bigr] \biggr\}.
\ee
Notice that $Z(0) = Z^{(u)}(\beta,\Lambda)$.

Let $i=1,\dots,d$ and $\ell = \frac12, \frac32, \dots, L-\frac12$. We let $R_{i,\ell}$ denote the reflection $\Lambda\to\Lambda$ across the edges $\{x,y\} \in \caE$ with $x_{i} = \ell-\frac12, y_{i} = \ell+\frac12$. That is,
\be
\label{def refl}
R_{i,\ell}(x_{1}, \dots, x_{i}, \dots, x_{d}) = (x_{1}, \dots, -(x_{i}-\ell), \dots, x_{d}),
\ee
where $-(x_{i}-\ell)$ is taken modulo $L$.
Let $\Lambda^{(1)}$ be the set of sites ``to the left'' of the reflexive plane, and $\Lambda^{(2)}$ the set of sites ``to its right''. Namely,
\be
\begin{split}
&\Lambda^{(1)} = \{ x\in\Lambda : x_{i} = \ell - \tfrac12 L, \dots, \ell-\tfrac12 \}, \\
&\Lambda^{(2)} = \{ x\in\Lambda : x_{i} = \ell + \tfrac12, \dots, \ell+\tfrac12 L \}.
\end{split}
\ee
We write $v = (v^{(1)},v^{(2)})$ where $v^{(i)} = v \big|_{\Lambda^{(i)}}$ is the restriction of $v$ on $\Lambda^{(i)}$, and $Rv^{(1)}$ is the field on $\Lambda^{(2)}$ such that
\be
(R v^{(1)})_{x} = v^{(1)}_{Rx},
\ee
for any $x \in \Lambda^{(2)}$.
Same for $v^{(2)}$. We can now formulate the property of reflection positivity.

\begin{proposition}
\label{prop rp}
Assume $u \in [0,\frac12]$ and $\theta = 2,3,4,\dots$
For any reflection $R$, we have
\[
Z(v^{(1)},v^{(2)})^{2} \leq Z(v^{(1)},Rv^{(1)}) \, Z(Rv^{(2)},v^{(2)}).
\]
\end{proposition}

\begin{proof}
Let $\vec v_{x} = v_{x} \vec e_{1}$, where $\vec e_{1}$ is the first unit vector in $\bbR^{2S}$. We rearrange the partition function \eqref{fpart fields} as follows.
\bm
Z(v) = \int\dd\rho_{\iota}(\xi) \sum_{\sigma\in\Sigma(\xi)} \exp \biggl\{ -\sum_{\{x,y\} \in \caE} \int_{0}^{\beta} \dd t \bigl( \vec\sigma_{xt} + \tfrac12 \vec v_{x} - \vec\sigma_{yt} - \tfrac12 \vec v_{y} \bigr)^{2} \\
+ \sum_{\{x,y\} \in \caE} \int_{0}^{\beta} \dd t (\vec\sigma_{xt} - \vec\sigma_{yt})^{2} \biggr\}.
\end{multline}
In order to reach a reflection positive form we split the Poisson process. Let the intensity $\iota'$ be defined by
\be
\begin{split}
&\iota' \Bigl( \Bigr\{ \transit{b}{a}{a}{b}, \transit{a}{b}{b}{a}, \transit{b}{b}{a}{a}, \transit{a}{a}{b}{b} \Bigr\} \Bigr) = u \qquad \text{if } a \neq b, \\
&\iota' \Bigl( \Bigr\{ \transit{b}{b}{a}{a} \Bigr\} \Bigr) = 1-2u \qquad \text{if } a \neq b,
\end{split}
\ee
and $\iota'(A)=0$ for all other subsets $A$ of $\{-S,\dots,S\}^{4}$. Let $\iota''$ be defined by
\be
\iota'' \Bigl( \Bigl\{ \transit{a}{a}{a}{a} \Bigr\}_{a \in \{-S,\dots,S\}} \Bigr) = 1
\ee
and $\iota''(A)=0$ on all other subsets $A$. Using Lemma \ref{lem add intensities} one verifies that
\bm
Z(v) = \int\dd\rho_{\iota'}(\xi') \int\dd\rho_{\iota''}(\xi'') \sum_{\sigma \in \Sigma(\xi' \cup \xi'')} \exp \biggl\{ -\sum_{\{x,y\} \in \caE} \int_{0}^{\beta} \dd t \bigl( \vec\sigma_{xt} + \tfrac12 \vec v_{x} - \vec\sigma_{yt} - \tfrac12 \vec v_{y} \bigr)^{2} \\
+ \sum_{\{x,y\} \in \caE} \int_{0}^{\beta} \dd t (\vec\sigma_{xt} - \vec\sigma_{yt})^{2} \biggr\}.
\end{multline}
Next, observe that given a realization $\xi'$ of $\rho_{\iota'}$, we have
\be
\int\dd\rho_{\iota''}(\xi'') \sum_{\sigma \in \Sigma(\xi' \cup \xi'')} F(\sigma) = \sum_{\sigma \in \Sigma(\xi')} F(\sigma) \int\dd\rho_{\iota''}(\xi'') \prod_{(x,y,t) \in \xi''} \delta_{\sigma_{xt}, \sigma_{yt}}.
\ee
This holds because space-time spin configurations are constant at the transitions of $\xi''$. Consequently, given $\xi'$ and $\sigma \in \Sigma(\xi')$,
\be
\begin{split}
\int\dd\rho_{\iota''}(\xi'') \prod_{(x,y,t) \in \xi''} \delta_{\sigma_{xt},\sigma_{yt}} &= \lim_{N\to\infty} \prod_{\{x,y\}\in\caE} \prod_{t=1}^{N} \Bigl( 1 - \frac\beta N + \frac\beta N \delta_{\sigma_{xt}, \sigma_{yt}} \Bigr) \\
&= \lim_{N\to\infty} \prod_{\{x,y\}\in\caE} \prod_{t=1}^{N} \Bigl( 1 - \frac\beta N (1 - \delta_{\sigma_{xt}, \sigma_{yt}}) \Bigr) \\
&= \exp\biggl\{ -\sum_{\{x,y\}\in\caE} \int_{0}^{\beta} \dd t \, (1 - \delta_{\sigma_{xt}, \sigma_{yt}}) \biggr\}.
\end{split}
\ee
We now use $1 - \delta_{\sigma_{xt}, \sigma_{yt}} = (\vec\sigma_{xt} - \vec\sigma_{yt})^{2}$ and we obtain
\be
Z(v) = \int\dd\rho_{\iota'}(\xi') \sum_{\sigma\in\Sigma(\xi')} \exp \biggl\{ -\sum_{\{x,y\} \in \caE} \int_{0}^{\beta} \dd t \bigl( \vec\sigma_{xt} + \tfrac12 \vec v_{x} - \vec\sigma_{yt} - \tfrac12 \vec v_{y} \bigr)^{2} \biggr\}.
\ee
The measure $\rho_{\iota'}$ is nonnegative if $u \in [0,\frac12]$, and it is reflection-symmetric.  See Fig.\ \ref{fig rp space} for an illustration.

\begin{centering}
\bfig
\includegraphics[width=5cm]{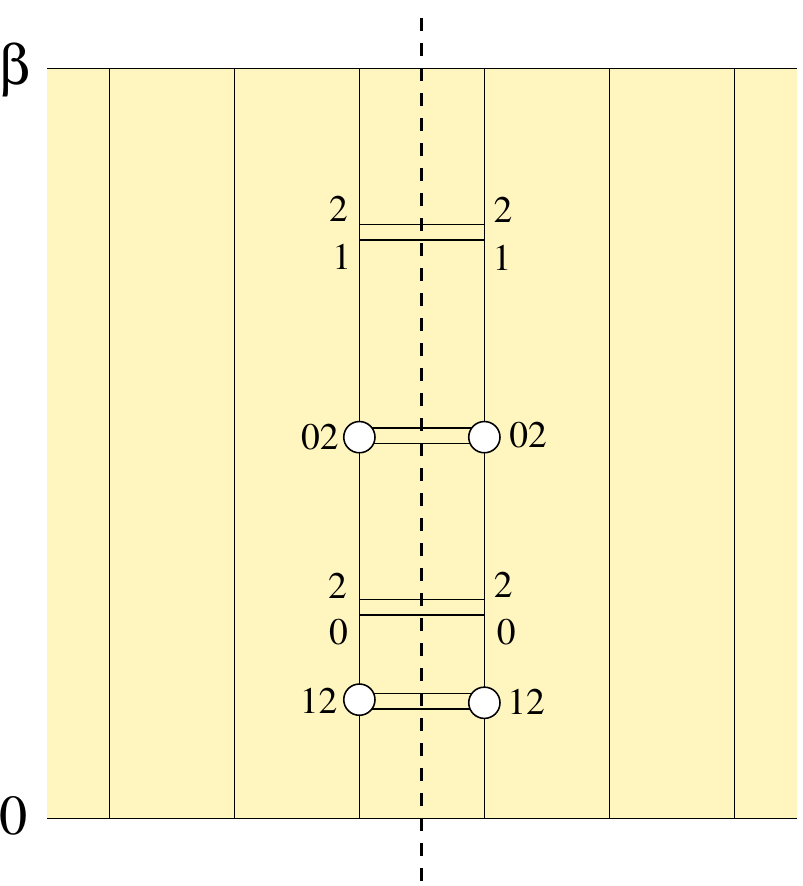}
\caption{Reflection in the space direction. We chose $S=2$ in this illustration. The white circles mean that the spin must flip from one prescribed value to the other, independently on both sides. It follows that the Poisson point process $\rho_{\iota'}$ is reflection symmetric since the constraints are identical in both halves.}
\label{fig rp space}
\efig
\end{centering}

Let us partition $\caE = \caE^{(1)} \cup \caE^{(2)} \cup \bar\caE$, where $\caE^{(1)}$ and $\caE^{(2)}$ are the sets of edges of $\Lambda^{(1)}$ and $\Lambda^{(2)}$, respectively, and $\bar\caE$ is the set of edges with endpoints in both $\Lambda^{(1)}$ and $\Lambda^{(2)}$. We can write
\be
\begin{split}
Z(v) = &\int_{\bar\caE \times [0,\beta]} \dd\rho_{\iota'}(\bar\xi) \\
&\biggl[ \int_{\caE^{(1)} \times [0,\beta]} \dd\rho_{\iota'}(\xi^{(1)}) \sum_{\sigma \in \Sigma_{\Lambda^{(1)}}(\xi^{(1)},\bar\xi)} \exp\Bigl\{ -\sum_{\{x,y\} \in \caE^{(1)}} \int_{0}^{\beta} \dd t \bigl( \vec\sigma_{xt} + \tfrac12 \vec v_{x} - \vec\sigma_{yt} - \tfrac12 \vec v_{y} \bigr)^{2} \Bigr\} \biggr] \\
&\biggl[ 1 \leftrightarrow 2 \biggr] \exp\Bigl\{ -\sum_{\{x,y\} \in \bar\caE} \int_{0}^{\beta} \dd t \bigl( \vec\sigma_{xt} + \tfrac12 \vec v_{x} - \vec\sigma_{yt} - \tfrac12 \vec v_{y} \bigr)^{2} \Bigr\}.
\end{split}
\ee
Here, $\Sigma_{\Lambda^{(1)}}(\xi^{(1)},\bar\xi)$ is the set of space-time configurations on $\Lambda^{(1)} \times [0,\beta]$ that are compatible with the specifications of $\xi^{(1)}$ and $\bar\xi$.

The couplings between both sides can be handled with the help of extra fields, as in \cite{FSS}. Namely, we have
\be
\begin{split}
\exp\biggl\{ &-\sum_{\{x,y\} \in \bar\caE} \int_{0}^{\beta} \dd t \bigl( \vec\sigma_{xt} + \tfrac12 \vec v_{x} - \vec\sigma_{yt} - \tfrac12 \vec v_{y} \bigr)^{2} \biggr\} \\
&= \lim_{N\to\infty} \prod_{\{x,y\} \in \bar\caE} \prod_{j=1}^{N} \exp\bigl\{ -\tfrac\beta N \bigl( \vec\sigma_{x,\frac{j\beta}N} + \tfrac12 \vec v_{x} - \vec\sigma_{y,\frac{j\beta}N} - \tfrac12 \vec v_{y} \bigr)^{2} \bigr\} \\
&=\lim_{N\to\infty} \prod_{\{x,y\} \in \bar\caE} \prod_{j=1}^{N} \int_{\bbR^{2S}} \dd\vec\alpha_{x,y,\frac{j\beta}N} \Bigl( \frac N{4\pi\beta} \Bigr)^{S} \exp\bigl\{ -\tfrac N{4\beta} \vec\alpha_{x,y,\frac{j\beta}N}^{2} \bigr\} \\
&\hspace{30mm} \exp\bigl\{ \ii \vec\alpha_{x,y,\frac{j\beta}N} \cdot (\vec\sigma_{x,\frac{j\beta}N} + \tfrac12 \vec v_{x} - \vec\sigma_{y,\frac{j\beta}N} - \tfrac12 v_{y}) \bigr\}.
\end{split}
\ee
Let us introduce
\be
\begin{split}
&F_{N}(v^{(1)}, \bar\xi, \{\vec\alpha_{x,y,t}\}) = \int_{\caE^{(1)} \times [0,\beta]} \dd\rho_{\iota'}(\xi^{(1)}) \sum_{\sigma \in \Sigma_{\Lambda^{(1)}}(\xi^{(1)},\bar\xi)} \\
&\exp\biggl\{ -\sum_{\{x,y\} \in \caE^{(1)}} \int_{0}^{\beta} \dd t \bigl( \vec\sigma_{xt} + \tfrac12 \vec v_{x} - \vec\sigma_{yt} - \tfrac12 \vec v_{y} \bigr)^{2} \biggr\} \exp\biggl\{ -\ii \sum_{\{x,y\} \in \bar\caE} \sum_{j=1}^{N} \vec\alpha_{x,y,\frac{j\beta}N} \cdot (\vec\sigma_{x,\frac{j\beta}N} + \tfrac12 \vec v_{x}) \biggr\}.
\end{split}
\ee
With this definition, the partition function $Z(v)$ takes the form
\bm
Z(v) = \lim_{N\to\infty} \int_{\bar\caE \times [0,\beta]} \dd\rho_{\iota'}(\bar\xi) \biggl( \prod_{\{x,y\} \in \bar\caE} \prod_{j=1}^{N} \int_{\bbR^{2S}} \dd\vec\alpha_{x,y,\frac{j\beta}N} \Bigl( \frac N{4\pi\beta} \Bigr)^{S} \e{-\frac N{4\beta} \vec\alpha_{x,y,\frac{j\beta}N}^{2}} \biggr) \\
F_{N}(v^{(1)}, \bar\xi, \{\vec\alpha_{x,y,t}\}) \overline{F_{N}}(v^{(2)}, \bar\xi, \{\vec\alpha_{x,y,t}\}).
\end{multline}
Using the Cauchy-Schwarz inequality, and retracing our steps backwards, we obtain the claim of the proposition.
\end{proof}

\begin{proposition}
\label{prop max Z(v)}
Assume that $u \in [0,\frac12]$ and $\theta = 2,3,\dots$ Then $Z(v)$ is maximized by $v \equiv 0$.
\end{proposition}

\begin{proof}
It is not hard to show that maximizers exist: $Z(v)$ is continuous and positive, we can fix one of the field values to be 0, and $Z(v^{(n)})$ tends to 0 along any sequence satisfying $\sup_{x} |v_{x}^{(n)}| \to \infty$ as $n\to\infty$, with $v_{0}^{(n)}=0$. The maximum can then be taken on a compact set.

The argument of the proof is then standard, one uses reflections to construct a sequence of maximizers where more and more fields are constant, until they are all identical. See Fig.\ \ref{fig rp maximize} for a quick illustration, and the references \cite{FSS,DLS,Bis} or the notes \cite{Toth2,Uel} for more details.
\end{proof}

\bfig
\includegraphics[width=130mm]{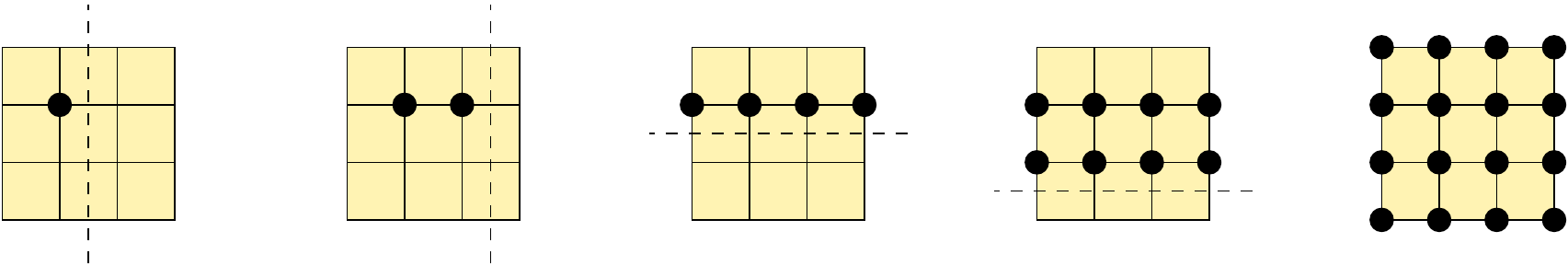}
\caption{Starting with a maximizer, reflections yield further maximizers where more and more values are identical.}
\label{fig rp maximize}
\efig

\subsection{Infrared bound for the correlation function}
\label{sec irb}

The name ``infrared bound'' refers to estimates of the Fourier transform of $\bbP(E_{0,x,0})$ around $k=0$, that is, for low ``frequencies''. We first use Corollary \ref{prop max Z(v)} in order to derive a bound on the Fourier transform of the Duhamel function. Then we use Falk-Bruch inequality to extend it to the ordinary correlation function. These steps follow \cite{DLS}.

Recall the definition of $\widetilde\kappa$ in Eq.\ \eqref{kappa FT}.

\begin{proposition}
\label{prop irb Duhamel}
Assume that $u \in [0,\frac12]$ and $\theta = 2,3,\dots$ For $k \in \Lambda^{*} \setminus \{0\}$, we have
\[
\widetilde\kappa(k,0) \leq \frac{2S+1}{\eps(k)}.
\]
\end{proposition}

From now on, we let $(v,v')$ denote the inner product in $\ell^{2}(\Lambda)$, that is,
\be
(v,v') = \sum_{x\in\Lambda} \bar v_{x} v_{x}'.
\ee

\begin{proof}
We have
\be
\begin{split}
\kappa(x,t) &= \frac1{Z^{(u)}(\beta,\Lambda)} \int\dd\rho_{u}(\omega) \sum_{\sigma\in\Sigma(\omega)} 1_{E_{0,x,t}}(\omega) \\
&= \frac1{Z^{(u)}(\beta,\Lambda)} \int\dd\rho_{u}(\omega) \sum_{\sigma\in\Sigma(\omega)} \frac{2S+1}S \vec\sigma_{00} \cdot \vec\sigma_{xt} \\
&=\frac{2S+1}S \bbE( \vec\sigma_{0,0} \cdot \vec\sigma_{x,t}).
\end{split}
\ee
It follows that
\be
\frac S{2S+1} \widetilde\kappa(k,0) = \sum_{x\in\Lambda} \int_{0}^{\beta} \dd t \e{\ii kx} \bbE( \vec\sigma_{00} \cdot \vec\sigma_{xt} ).
\ee
Let $\Delta$ denote the discrete Laplacian, such that
\be
(\Delta v)_{x} = \sum_{y: \{x,y\} \in \caE} (v_{y} - v_{x}).
\ee
We have
\be
Z(v) = \int\dd\rho_{\iota}(\xi) \sum_{\sigma\in\Sigma(\xi)} \exp \biggl\{ \int_{0}^{\beta} \dd t (\vec\sigma_{\cdot,t}^{(1)}, \Delta v) + \frac\beta4 (v, \Delta v) \biggr\}.
\ee
We choose $v = \cos(kx)$. Expanding around $v=0$ to second order, using $-\Delta v = \eps(k) v$, we get
\be
\begin{split}
Z(\eta v) &= Z(0) + \int\dd\rho_{\iota}(\xi) \sum_{\sigma\in\Sigma(\xi)} \biggl[ \tfrac12 \eta^{2} \int_{0}^{\beta} \dd t \int_{0}^{\beta} \dd t' \eps(k)^{2} (\vec\sigma^{(1)}_{\cdot,t}, v) (\vec\sigma^{(1)}_{\cdot,t'}, v) - \tfrac14 \eta^{2} \beta \eps(k) (v,v) \biggr] + O(\eta^{4})\\
&= Z(0) \biggl[ 1 + \tfrac12 \eta^{2} \beta \eps(k)^{2} \int_{0}^{\beta} \dd t \; \bbE \Bigl( (\vec\sigma^{(1)}_{\cdot,0}, v) (\vec\sigma^{(1)}_{\cdot,t}, v) \Bigr) - \tfrac14 \eta^{2} \beta \eps(k) \biggr] + O(\eta^{4}).
\end{split}
\ee
We calculate the expectation:
\be
\begin{split}
\bbE \Bigl( (\vec\sigma^{(1)}_{\cdot,0}, v) (\vec\sigma^{(1)}_{\cdot,t}, v) \Bigr) &= \sum_{x,y\in\Lambda} \cos kx \cos ky \; \bbE \Bigl( \vec\sigma_{x0}^{(1)} \vec\sigma_{yt}^{(1)} \Bigr) \\
&= \sum_{x,z\in\Lambda} \cos kx \cos k(x-z) \; \bbE \Bigl( \vec\sigma_{00}^{(1)} \vec\sigma_{zt}^{(1)} \Bigr).
\end{split}
\ee
The last line is obtained by replacing $y$ by $-y$, by using translation invariance, and with the substitution $z = x+y$. Observe now that
\be
\bbE( \vec\sigma_{00}^{(1)} \vec\sigma_{zt}^{(1)} ) = \frac1{2S} \bbE ( \vec\sigma_{0,0} \cdot \vec\sigma_{zt} ) = \frac1{2 (2S+1)} \kappa(z,t).
\ee
We now get the Fourier transform of correlations:
\be
\begin{split}
\sum_{z\in\Lambda} \cos k(x-z) \bbE\Bigl( \vec\sigma_{00}^{(1)} \vec\sigma_{zt}^{(1)} \Bigr) &= \frac1{2 (2S+1)} \Re \sum_{z\in\Lambda} \e{\ii kx} \e{-\ii kz} \kappa(z,t) \\
&= \frac1{2 (2S+1)} \Re \e{\ii kx} \widehat\kappa(k,t) \\
&= \frac1{2 (2S+1)} \cos kx \; \widehat\kappa(k,t).
\end{split}
\ee
The last identity holds because $\widehat\kappa(k,t)$ is real due to lattice symmetries. Then
\be
Z(\eta v) = Z(0) (v,v) \biggl[ 1 + \tfrac1{4(2S+1)} \eta^{2} \beta \eps(k)^{2} \int_{0}^{\beta} \dd t \; \widehat\kappa(k,t) - \tfrac14 \eta^{2} \beta \eps(k) \biggr] + O(\eta^{4}).
\ee
The bracket is negative for small $\eta$, so that
\be
\frac1{2S+1} \eps(k) \int_{0}^{\beta} \dd t \; \widehat\kappa(k,t) \leq 1.
\ee
The proposition follows.
\end{proof}

The next step is to transfer this bound to the Fourier transform $\widehat\kappa(k,0)$, see Eq.\ \eqref{kappa FT}. We need to compute the double commutator of the Falk-Bruch inequality.
Let us introduce the operators $Q_{xy}^{33}$ and $T_{xy}^{33}$ in $\caH_{x,y}$ by
\be
\label{def commutator operators}
\begin{split}
&\langle a,b| Q_{xy}^{33} |c,d\rangle = (a-c)^{2} \delta_{ab} \delta_{cd}, \\
&\langle a,b| T_{xy}^{33} |c,d\rangle = (a-c)^{2} \delta_{ad} \delta_{bc}.
\end{split}
\ee
The peculiar notation is motivated by the fact that these operators are given by double commutators with $S_{x}^{3}$.

\begin{lemma}
\label{lem Q33 T33}
\[
\begin{split}
&Q_{xy}^{33} = -[S_{x}^{3}, [Q_{xy},S_{x}^{3}]] = -[S_{y}^{3}, [Q_{xy},S_{x}^{3}]], \\
&T_{xy}^{33} = -[S_{x}^{3}, [T_{xy},S_{x}^{3}]] = [S_{y}^{3}, [T_{xy},S_{x}^{3}]] = (S_{x}^{3} - S_{y}^{3})^{2} T_{xy}.
\end{split}
\]
\end{lemma}

The proof of the lemma is immediate by looking at the action of these operators on basis elements. Next we consider
\be
\widehat S_{k}^{3} = \sum_{x\in\Lambda} \e{-\ii kx} S_{x}^{3},
\ee
and we introduce
\be
\begin{split}
\tau_{1}^{(u)} &= \langle T_{0,e_{i}}^{33} \rangle, \\
\tau_{0}^{(u)} &= \langle Q_{0,e_{i}}^{33} \rangle,
\end{split}
\ee
where the expectations are taken in the Gibbs state with Hamiltonian $H_{\Lambda}^{(u)}$. These numbers turn out to be related to the probabilities of $E_{0,e_{1},0}^{\pm}$, see Lemma \ref{lem bornes tau} below.

\begin{lemma}
\label{lem double commutateur}
 \[
\langle [\widehat S_{-k}^{3}, [H_{\Lambda}^{(u)}, \widehat S_{k}^{3}]] \rangle = |\Lambda| \bigl( u \eps(k) \tau_{1}^{(u)} + (1-u) \eps(k+\pi) \tau_{0}^{(u)} \bigr).
\]
\end{lemma}

\begin{proof}
We perform the computations for $H^{(0)}_{\Lambda}$ and $H^{(1)}_{\Lambda}$ separately, and we use linearity to get the result for $H_{\Lambda}^{(u)}$.
\be
\begin{split}
[H_{\Lambda}^{(1)}, \widehat S_{k}^{3}] &= \sum_{x\in\Lambda} \e{-\ii kx} [H_{\Lambda}^{(1)}, S_{x}^{3}] \\
&= -\sum_{x,y : \{x,y\} \in \caE} \e{-\ii kx} [T_{xy}, S_{x}^{3}] \\
&= -\sum_{x,y : \{x,y\} \in \caE} \e{-\ii kx} T_{xy} (S_{x}^{3} - S_{y}^{3}).
\end{split}
\ee
\be
\begin{split}
[\widehat S_{-k}^{3}, [H_{\Lambda}^{(1)}, \widehat S_{k}^{3}]] &= -\sum_{x,y : \{x,y\} \in \caE} \e{-\ii kx} [ \e{\ii kx} S_{x}^{3} + \e{\ii ky} S_{y}^{3}, T_{xy} (S_{x}^{3} - S_{y}^{3})] \\
&= -\sum_{x,y : \{x,y\} \in \caE} [ S_{x}^{3} + \e{-\ii k(x-y)} S_{y}^{3}, T_{xy} (S_{x}^{3} - S_{y}^{3})] \\
&= -\sum_{x,y : \{x,y\} \in \caE} \bigl( 1 - \e{-\ii k(x-y)} \bigr) (S_{x}^{3} - S_{y}^{3})^{2} T_{xy} \\
&= 2\sum_{\{x,y\} \in \caE} \bigl( 1 - \cos k(x-y) \bigr) (S_{x}^{3} - S_{y}^{3})^{2} T_{xy}.
\end{split}
\ee
Next, we perform the calculations for $H_{\Lambda}^{(0)}$. For the commutator, we get
\be
[H_{\Lambda}^{(0)}, \widehat S_{k}^{3}] = \sum_{x\in\Lambda} \e{-\ii kx} [H_{\Lambda}^{(0)}, S_{x}^{3}] = -\sum_{x,y : \{x,y\} \in \caE} \e{-\ii kx} [Q_{xy}, S_{x}^{3}].
\ee
And for the double commutator, we get
\be
[\widehat S_{-k}^{3}, [H_{\Lambda}^{(0)}, \widehat S_{k}^{3}]] = -\sum_{x,y : \{x,y\} \in \caE} [ S_{x}^{3} + \e{-\ii k(x-y)} S_{y}^{3}, Q_{xy} (S_{x}^{3} - S_{y}^{3})].
\ee
The result follows from Lemma \ref{lem Q33 T33}.
\end{proof}

\begin{proposition}
\label{prop irb}
\[
\widehat{\langle S_{0}^{3} S_{x}^{3} \rangle}(k) \leq \tfrac12 \sqrt{\tfrac{S(S+1)(2S+1)}3} \sqrt{u\tau_{1}^{(u)} + (1-u) \tau_{0}^{(u)} \frac{\eps(k+\pi)}{\eps(k)}} + \frac{S(S+1)(2S+1)}{3 \beta \eps(k)}.
\]
\end{proposition}

\begin{proof}
We use the inequality \eqref{FB} with
\be
\langle A^{*} A + A A^{*}\rangle = \langle \widehat S_{-k}^{3} \widehat S_{k}^{3} + \widehat S_{k}^{3} \widehat S_{-k}^{3} \rangle = 2|\Lambda| \widehat{\langle S_{0}^{3} S_{x}^{3} \rangle}(k).
\ee
We have
\be
{(\hat S_{k}^{3}, \hat S_{k}^{3})_{\rm Duh}} = \tfrac13 S(S+1) |\Lambda| \widetilde\kappa(k,0).
\ee
Using Lemma \ref{lem double commutateur} we get
\bm
2|\Lambda| \widehat{\langle S_{0}^{3} S_{x}^{3} \rangle}(k) \leq \sqrt{\tfrac13 S(S+1) |\Lambda| \widetilde\kappa(k,0)} \sqrt{|\Lambda| \bigl( u \eps(k) \tau_{1}^{(u)} + (1-u) \eps(k+\pi) \tau_{0}^{(u)} \bigr)} \\
+ \frac{2 S(S+1)}{3\beta} |\Lambda| \widetilde\kappa(k,0).
\end{multline}
The claim now follows from Proposition \ref{prop irb Duhamel}.
\end{proof}

The next lemma shows that $\tau_{0}^{(u)}$ and $\tau_{1}^{(u)}$ are related to $\bbP(E^{\pm}_{0,e_{1},0})$.

\begin{lemma}
\label{lem bornes tau}
For any $u \in [0,1]$, we have
\begin{itemize}
\item[(a)] $\tau_{0}^{(u)} = \frac23 S(S+1)(2S+1) \bbP(E_{0,e_{1},0}^{-})$.
\item[(b)] $\tau_{1}^{(u)} = \frac23 S(S+1)(2S+1) \bbP(E_{0,e_{1},0}^{+})$.
\end{itemize}
\end{lemma}

\begin{proof}
We start with (a). Given a realization $\omega$ of the Poisson point process $\rho_{u}$, we denote $\bar\omega$ the realization where we add a double bar on the edge $(0,e_{1})$ at time $t=0$. We have
\be
\label{expression tau0}
\begin{split}
\tau_{0}^{(u)} = \langle Q_{0e_{1}}^{33} \rangle &= \frac1{Z^{(u)}(\beta,\Lambda)} \int\dd\rho_{u}(\omega) \sum_{\sigma \in \Sigma(\bar\omega)} (\sigma_{0,0+} - \sigma_{0,0-})^{2} \\
&= \frac1{Z^{(u)}(\beta,\Lambda)} \biggl[ \int_{E_{0,e_{1},0}^{\rm c}} \dd\rho_{u}(\omega) \sum_{\sigma \in \Sigma(\bar\omega)} (\sigma_{0,0+} - \sigma_{0,0-})^{2} \\
&\hspace{19mm} + \int_{E_{0,e_{1},0}} \dd\rho_{u}(\omega) \sum_{\sigma \in \Sigma(\bar\omega)} (\sigma_{0,0+} - \sigma_{0,0-})^{2} \biggr].
\end{split}
\ee
The first integral is over realizations $\omega$ such that $(0,0)$ and $(e_{1},0)$ belong to different loops. The additional double bar merges these two loops, so that $\bar\omega \in E_{0,e_{1},0}$ and $\sigma_{0,0+} = \sigma_{0,0-}$ for all compatible configurations. The first integral is then zero.

The second integral is over $\omega \in E_{0,e_{1},0}$. In this case, $\bar\omega \in E_{0,e_{1},0}$ if $\omega \in E_{0,e_{1},0}^{+}$ and $\bar\omega \in E_{0,e_{1},0}^{\rm c}$ if $\omega \in E_{0,e_{1},0}^{-}$. We get 0 in the first case, and
\be
\begin{split}
\sum_{\sigma\in\Sigma(\bar\omega))} (\sigma_{0,0+} - \sigma_{0,0-})^{2} &= (2S+1)^{|\caL(\omega)|-1} \sum_{a,b=-S}^{S} (a-b)^{2} \\
&= (2S+1)^{|\caL(\omega)|} \tfrac23 S(S+1)(2S+1)
\end{split}
\ee
in the second case. We used the identity
\be
\frac1{2S+1} \sum_{a,b=-S}^{S} (a-b)^{2} = \tfrac23 S(S+1)(2S+1)
\ee
which holds for all $S \in \frac12 \bbN$.

The proof of (b) is very similar. We have
\be
\tau_{1}^{(u)} = \frac1{Z^{(u)}(\beta,\Lambda)} \int\dd\rho_{u}(\omega) 1_{E_{0,e_{1},0}}(\omega)  1_{E_{0,e_{1},0}^{\rm c}}(\tilde\omega) \sum_{\sigma\in\Sigma(\tilde\omega)} (\sigma_{0,0+} - \sigma_{0,0-})^{2}.
\ee
It differs from \eqref{expression tau0} in that $\tilde\omega$ contains an extra cross instead of a double bar. The product of the two indicators is equal to $1_{E_{0,e_{1},0}^{+}}(\omega)$ and we get the claim of the lemma.
\end{proof}

Let us introduce the function $\alpha(u)$ with values in $[0,1]$:
\be
\bbP(E^{+}_{0,e_{1},0}) = \alpha(u) \, \bbP(E_{0,e_{1},0}).
\ee
Using Lemma \ref{lem bornes tau}, the infrared bound of Proposition \ref{prop irb} can be written as
\bm
\label{borne infrarouge}
\widehat{\langle S_{0}^{3} S_{x}^{3} \rangle}(k) \leq \frac{S(S+1)(2S+1)}{3\sqrt2} \sqrt{\bbP(E_{0,e_{1},0})} \sqrt{u \alpha(u) + (1-u) (1-\alpha(u)) \frac{\eps(k+\pi)}{\eps(k)}} \\
+ \frac{S(S+1)(2S+1)}{3 \beta \eps(k)}.
\end{multline}

We now prove the claim about the occurrence of macroscopic loops.

\begin{proof}[Proof of Theorem \ref{thm irb}]
We have
\be
\langle (S_{0}^{3})^{2} \rangle = \frac1{Z^{(u)}(\beta,\Lambda)} \int\dd\rho_{u}(\omega) \sum_{\sigma\in\Sigma(\omega)} \sigma_{00}^{2} = \tfrac13 S(S+1).
\ee
Then
\be
\tfrac13 S(S+1) = \langle S_{0}^{3} S_{x}^{3} \rangle \Big|_{x=0} = \frac1{|\Lambda|} \widehat{\langle S_{0}^{3} S_{x}^{3} \rangle}(0) + \frac1{|\Lambda|} \sum_{k \in \Lambda^{*} \setminus\{0\}} \widehat{\langle S_{0}^{3} S_{x}^{3} \rangle}(k).
\ee
Using the infrared bound of Eq.\ \eqref{borne infrarouge}, we obtain
\begin{multline}
\label{la borne infrarouge}
\frac1{|\Lambda|} \sum_{x\in\Lambda} \langle S_{0}^{3} S_{x}^{3} \rangle \geq \tfrac13 S(S+1) \\
- \tfrac{S(S+1)(2S+1)}{3\sqrt2} \sqrt{\bbP(E_{0,e_{1},0})} \frac1{|\Lambda|} \sum_{k \in \Lambda^{*} \setminus\{0\}} \sqrt{u \alpha(u) + (1-u) (1-\alpha(u)) \frac{\eps(k+\pi)}{\eps(k)}} \\
- \tfrac{S(S+1)(2S+1)}{3} \frac1{\beta |\Lambda|} \sum_{k \in \Lambda^{*} \setminus\{0\}} \frac1{\eps(k)}.
\end{multline}
Taking $L\to\infty$ then $\beta\to\infty$, we obtain
\be
\lim_{\beta\to\infty} \lim_{L\to\infty} \frac1{|\Lambda|} \sum_{x\in\Lambda} \langle S_{0}^{3} S_{x}^{3} \rangle \geq \tfrac13 S(S+1) - \tfrac{S(S+1)(2S+1)}{3\sqrt2} \sqrt{\bbP(E_{0,e_{1},0})} \, J_{d}^{u,\alpha(u)}.
\ee
where
\be
J_{d}^{u,\alpha} = \frac1{(2\pi)^{d}} \int_{[-\pi,\pi]^{d}} \sqrt{u \alpha + (1-u) (1-\alpha) \frac{\eps(k+\pi)}{\eps(k)}} \dd k.
\ee
By differentiating twice, we can verify that $J_{d}^{u,\alpha}$ is concave with respect to $\alpha$, and that its derivative at $\alpha=0$ is equal to
\be
\frac{\partial}{\partial\alpha} J_{d}^{u,\alpha} \Big|_{\alpha=0} = \frac{(2u-1) d}{\sqrt{1-u}} \frac1{(2\pi)^{d}} \int_{[-\pi,\pi]^{d}} \frac{\dd k}{\sqrt{\eps(k) \eps(k+\pi)}},
\ee
which is negative for $u \in [0,\frac12]$. Then
\be
J_{d}^{u,\alpha(u)} \leq J_{d}^{u,0} = \sqrt{1-u} \, J_{d}
\ee
and the first lower bound of Theorem \ref{thm irb} follows immediately.

The second lower bound is obtained using the sum rule suggested in \cite{KLS1},
\be
\begin{split}
\frac1{|\Lambda|} \sum_{k \in \Lambda^{*}} \widehat{\langle S_{0}^{3} S_{x}^{3} \rangle}(k) \cos k_{1} &= \frac1{2|\Lambda|} \sum_{k \in \Lambda^{*}} \widehat{\langle S_{0}^{3} S_{x}^{3} \rangle}(k) \bigl( \e{\ii k_{1}} + \e{-\ii k_{1}} \bigr) \\
&= \tfrac12 \bigl( \langle S_{0}^{3} S_{e_{1}}^{3} \rangle + \langle S_{0}^{3} S_{-e_{1}}^{3} \rangle \bigr) \\
&= \tfrac13 S(S+1) \bbP(E_{0,e_{1},0}).
\end{split}
\ee
Isolating the term $k=0$ and using Eq.\ \eqref{borne infrarouge}, we get
\be
\begin{split}
&\frac1{|\Lambda|} \sum_{x\in\Lambda} \langle S_{0}^{3} S_{x}^{3} \rangle \geq \tfrac13 S(S+1) \bbP(E_{0,e_{1},0}) - \frac1{|\Lambda|} \sum_{k \in \Lambda^{*} \setminus \{0\}} \widehat{\langle S_{0}^{3} S_{x}^{3} \rangle}(k) \frac1d \sum_{i=1}^{d} \cos k_{i} \\
&\geq \tfrac13 S(S+1) \bbP(E_{0,e_{1},0}) \\
&- \tfrac{S(S+1)(2S+1)}{3\sqrt2} \sqrt{\bbP(E_{0,e_{1},0})} \frac1{|\Lambda|} \sum_{k \in \Lambda^{*} \setminus \{0\}} \sqrt{u \alpha(u) + (1-u) (1-\alpha(u)) \frac{\eps(k+\pi)}{\eps(k)}} \Bigl( \frac1d \sum_{i=1}^{d} \cos k_{i} \Bigr)_{+} \\
&- \tfrac{S(S+1)(2S+1)}{3} \frac1{\beta |\Lambda|} \sum_{k \in \Lambda^{*} \setminus\{0\}} \frac1{\eps(k)}.
\end{split}
\ee
Taking $L\to\infty$ then $\beta\to\infty$, we obtain
\be
\lim_{\beta\to\infty} \lim_{L\to\infty} \frac1{|\Lambda|} \sum_{x\in\Lambda} \langle S_{0}^{3} S_{x}^{3} \rangle \geq \tfrac13 S(S+1) \bbP(E_{0,e_{1},0}) - \tfrac{S(S+1)(2S+1)}{3\sqrt2} \sqrt{\bbP(E_{0,e_{1},0})} \, I_{d}^{u,\alpha(u)}.
\ee
where
\be
I_{d}^{u,\alpha} = \frac1{(2\pi)^{d}} \int_{[-\pi,\pi]^{d}} \sqrt{u \alpha + (1-u) (1-\alpha) \frac{\eps(k+\pi)}{\eps(k)}}\Bigl( \frac1d \sum_{i=1}^{d} \cos k_{i} \Bigr)_{+} \dd k.
\ee
The derivative of $I_{d}^{u,\alpha}$ with respect to $\alpha$ is negative when $u \in [0,\frac12]$ and therefore
\be
I_{d}^{u,\alpha(u)} \leq I_{d}^{u,0} = \sqrt{1-u} \, I_{d}.
\ee
The second lower bound of Theorem \ref{thm irb} follows.
\end{proof}

\section{Reflection positivity in space and time}
\label{sec sp rp}

This section describes an extension of the method of reflection positivity and infrared bounds to the space-time. As it turns out, the results are not as good as in the former section. It seems nonetheless useful to include this section since the idea is natural and it may possibly be improved in the future. Reflection positivity in space and time has been independently (and indeed, earlier) proposed by Bj\"ornberg for the quantum Ising model in order to prove interesting results about critical exponents in $d\geq3$ \cite{Bjo}. The method described here shares many similarities.

In order to state the main result is the following, we need the following integrals, that are similar to those of Eq.\ \eqref{IJ}:
\be
\begin{split}
&I_{d}' = \frac1{(2\pi)^{d}} \int_{[-\pi,\pi]^{d}} \sqrt{\frac{2d}{\eps(k)}} \Bigl( \frac1d \sum_{i=1}^{d} \cos k_{i} \Bigr)_{+} \dd k, \\
&J_{d}' = \frac1{(2\pi)^{d}} \int_{[-\pi,\pi]^{d}} \sqrt{\frac{2d}{\eps(k)}} \dd k.
\end{split}
\ee

\begin{theorem}
\label{thm gen irb}
Let $d\geq3$ and $u \in [0,\frac12]$. We have the two lower bounds
\[
\lim_{\beta\to\infty} \lim_{L\to\infty} \bbE \Bigl( \frac{L_{(0,0)}}{\beta L^{d}} \Bigr) \geq \begin{cases} 1 - (2S+1) \sqrt{1-u} \, J_{d}' \sqrt{\bbP(E_{0,e_{1},0})}; \\ \bbP(E_{0,e_{1},0}) - (2S+1) \sqrt{1-u } \, I_{d}' \sqrt{\bbP(E_{0,e_{1},0})}. \end{cases}
\]
\end{theorem}

One finds numerically that $I_{2}' = 0.489$, $I_{3}' = 0.278$, $J_{2}' = 1.286$, $J_{3}' = 1.115$, etc..., and these numbers are larger than $I_{d}/\sqrt2$ and $J_{d}/\sqrt2$. Theorem \ref{thm gen irb} is no improvement of Theorem \ref{thm irb}, disappointingly.

The rest of the section is devoted to proving this theorem. The method of proof does not rely on analytic inequalities such as Falk-Bruch, and it will be more attractive to some readers. Its difficulty is about the same, though.

We generalize the notion of partition function with external fields, by considering fields $v : \Lambda \times [0,\beta]_{\rm per} \to \bbR$. That is, $v$ also depends on the ``time'' parameter. We prove that the partition function has a local maximum at $v=0$, and we obtain a generalized infrared bound for the Fourier transform in space {\and} time of the correlation function.

Let $\caV_{c_{0}}$ be the set of fields $v : \Lambda \times [0,\beta]_{\rm per} \to \bbR$ where $v_{xt}$ is twice differentiable with respect to $t$, and $\bigl| \frac{\partial v_{x,t}}{\partial t} \bigr| \leq c_{0}$ for every $x,t$. We introduce the partition function by
\bm
\label{fpart gen fields}
Z(v) = \int\dd\rho_{u}(\omega) \sum_{\sigma\in\Sigma(\omega)} \exp \biggl\{ -\sum_{\{x,y\} \in \caE} \int_{0}^{\beta} \dd t \Bigl[ (\vec\sigma_{xt}^{(1)} - \vec\sigma_{yt}^{(1)}) (v_{xt}-v_{yt}) + \tfrac14 (v_{xt}-v_{yt})^{2} \Bigr] \\
+ \sum_{x\in\Lambda} \int_{0}^{\beta} \dd t \Bigl[ a \vec\sigma_{xt}^{(1)} \frac{\partial^{2} v_{xt}}{\partial t^{2}} - b \Bigl( \frac{\partial v_{xt}}{\partial t} \Bigr)^{2} \biggr\}.
\end{multline}
The constants $a$ and $b$ will be chosen later.

\begin{proposition}
\label{prop space-time rp}
Assume that $u \in [0,\frac12]$. For every $v \in \caV_{c_{0}}$, there exists $v^{*} = (v^{*}_{t}) \in \caV_{c_{0}}$ that depends on $t$ but not on $x$, such that $Z(v) \leq Z(v^{*})$.
\end{proposition}

The proof can be done by extending Proposition \ref{prop rp} to the partition function above, and by repeating the proof of Proposition \ref{prop max Z(v)}. It turns out that the time-dependence of the fields and the extra terms do not play any r\^ole.

\begin{proposition}
\label{prop gen GD}
Assume that $u \in [0,\frac12]$ and that $b > 2d a^{2} (1-u) \kappa(e_{1},0)$. Then there exists $c_{0}>0$ such that $Z(v) \leq Z(0)$ for every $v \in \caV_{c_{0}}$.
\end{proposition}

\begin{proof}
Because of Proposition \ref{prop space-time rp} it is enough to prove it for space-invariant fields only. We must show that maximizers are time-invariant as well. We have
\be
Z(v) = \lim_{N\to\infty} Z_{N}(v)
\ee
where
\bm
Z_{N}(v) = \int\dd\rho_{u}(\omega) \sum_{\sigma\in\Sigma(\omega)} \exp \biggl\{ -\frac N\beta \sum_{x\in\Lambda} \sum_{t \in \frac\beta N \{1,\dots,N\}} \\
\Bigl[ a (\vec\sigma_{x,t+\frac\beta N}^{(1)} - \vec\sigma_{x,t}^{(1)}) (v_{t+\frac\beta N} - v_{t}) + b (v_{t+\frac\beta N} - v_{t})^{2} \biggr\}.
\end{multline}
Indeed, we have discretized $\frac{\partial v_{t}}{\partial t}$ and $\frac{\partial^{2} v_{t}}{\partial t^{2}}$ and used the discrete integration by parts. We take the limit along $N \in 2\bbN$.

We apply ``horizontal reflections'' across the planes determined by $t = n\frac\beta N$, $n=1,\dots,N$. Then $Z_{N}(v^{(1)},v^{(2)}) \leq Z_{N}(v^{(1)},Rv^{(1)}) Z_{N}(Rv^{(2)},v^{(2)})$ from the Cauchy-Schwarz inequality --- the situation is actually simpler than in the proof of Proposition \ref{prop rp} as we do not need to introduce extra fields to decouple the two parts. See Fig.\ \ref{fig rp time} for an illustration. Proceeding as before, we obtain the existence of a maximizer $v^{*}$ for $Z_{N}$ with jigsaw shape, namely
\be
v^{*}_{t} = (-1)^{\frac{Nt}\beta} \frac cN
\ee
for some constant $|c| \leq c_{0}\beta/2$.

\begin{centering}
\bfig
\begin{picture}(0,0)%
\includegraphics{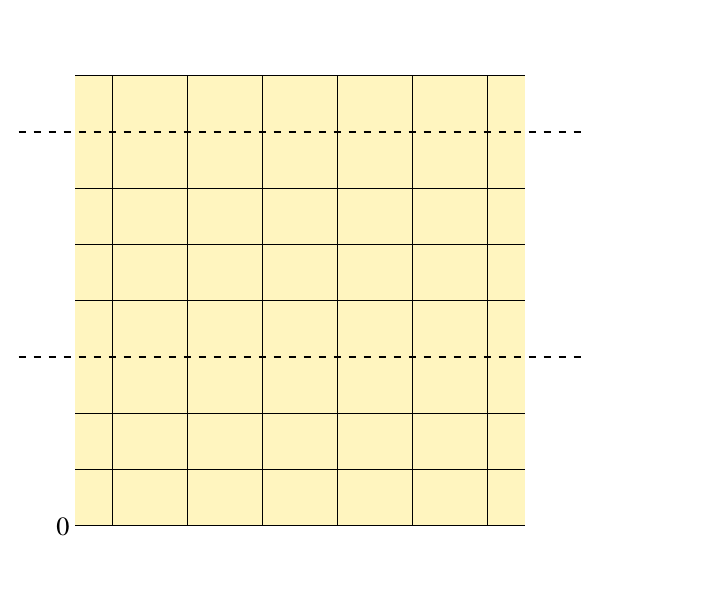}
\end{picture}%
\setlength{\unitlength}{2368sp}%
\begingroup\makeatletter\ifx\SetFigFont\undefined%
\gdef\SetFigFont#1#2#3#4#5{%
  \reset@font\fontsize{#1}{#2pt}%
  \fontfamily{#3}\fontseries{#4}\fontshape{#5}%
  \selectfont}%
\fi\endgroup%
\begin{picture}(5645,4802)(600,-4562)
\put(976,-436){\makebox(0,0)[lb]{\smash{{\SetFigFont{8}{9.6}{\rmdefault}{\mddefault}{\updefault}{\color[rgb]{0,0,0}$\beta$}%
}}}}
\put(4876,-1786){\makebox(0,0)[lb]{\smash{{\SetFigFont{8}{9.6}{\familydefault}{\mddefault}{\updefault}{\color[rgb]{0,0,0}$(Rv^{(1)})_t$}%
}}}}
\put(4876,-3586){\makebox(0,0)[lb]{\smash{{\SetFigFont{8}{9.6}{\familydefault}{\mddefault}{\updefault}{\color[rgb]{0,0,0}$v_t^{(1)}$}%
}}}}
\put(601,-3586){\makebox(0,0)[lb]{\smash{{\SetFigFont{8}{9.6}{\rmdefault}{\mddefault}{\itdefault}{\color[rgb]{0,0,0}$\beta/N$}%
}}}}
\end{picture}%

\caption{Reflection in the time direction, across the dotted lines. There are two dotted lines because of periodicity.}
\label{fig rp time}
\efig
\end{centering}

We need to show that $c=0$. We have
\be
Z_{N}(v^{*}) = \int\dd\rho_{u}(\omega) \sum_{\sigma\in\Sigma(\omega)} \exp \biggl\{ -\frac{4bc^{2} |\Lambda|}\beta - \frac{2ac}\beta \sum_{x \in \Lambda} \sum_{t \in \frac\beta N \{1,\dots,N\}} (-1)^{\frac{Nt}\beta} \bigl( \vec\sigma_{x,t+\frac\beta N}^{(1)} - \vec\sigma_{xt}^{(1)} \bigr) \biggr\}.
\ee
We need to show that the term linear in $c$ is actually quadratic. This can be done using the following trick. When integrating over realizations $\omega$ with the measure $\rho_{u}$, we replace each transition $(x,y,t) \in \omega$ by $\frac12$ transition at $t$, and $\frac12$ transition at $t + \frac\beta N$. Because of the alternating sign in the maximizer $v^{*}$, we obtain
\bm
Z_{N}(v^{*}) = \e{-\frac{4bc^{2}}\beta |\Lambda|} \int\dd\rho_{u}(\omega) \sum_{\sigma\in\omega} \prod_{(x,y,t) \in \omega} \tfrac12 \Biggl( \exp\biggl\{ \frac{2ac}\beta \bigl( \vec\sigma_{x,t+\frac\beta N}^{(1)} + \vec\sigma_{y,t+\frac\beta N}^{(1)} - \vec\sigma_{xt}^{(1)} - \vec\sigma_{yt}^{(1)} \bigr) \biggr\} \\
+ \exp\biggl\{ -\frac{2ac}\beta \bigl( \vec\sigma_{x,t+\frac\beta N}^{(1)} + \vec\sigma_{y,t+\frac\beta N}^{(1)} - \vec\sigma_{xt}^{(1)} - \vec\sigma_{yt}^{(1)} \bigr) \biggr\} \Biggr) + O\Bigl( \frac1N \Bigr).
\end{multline}
The correction $O(\frac1N)$ is due to the realizations $\omega$ where transitions occur at almost the same location and time.

If the transition is a cross we have $\vec\sigma_{x,t+\frac\beta N} = \vec\sigma_{y,t}$ and $\vec\sigma_{y,t+\frac\beta N} = \vec\sigma_{x,t}$ and the corresponding factor is 1. If the transition is a double bar, we get the factor
\be
\cosh\Bigl[ \frac{4ac}\beta \bigl( \vec\sigma_{x,t+\frac\beta N}^{(1)} - \vec\sigma_{xt}^{(1)} \bigr) \Bigr] = \exp\Bigl\{ \frac{8a^{2}c^{2}}{\beta^{2}} \bigl( \vec\sigma_{x,t+\frac\beta N}^{(1)} - \vec\sigma_{xt}^{(1)} \bigr)^{2} + O(c^{4}) \Bigr\}.
\ee
Let $\caT(\omega)$ denote the set of double bars that are present in the realization $\omega$.
We then obtain
\be
Z_{N}(v^{*}) = \e{-\frac{4bc^{2}}\beta |\Lambda|} \int\dd\rho_{u}(\omega) \sum_{\sigma\in\omega} \exp\Bigl\{ \tfrac{8a^{2}c^{2}}{\beta^{2}} \sum_{(x,y,t) \in \caT(\omega)} \bigl( \vec\sigma_{x,t+\frac\beta N}^{(1)} - \vec\sigma_{xt}^{(1)} \bigr)^{2} + O(c^{4}) \Bigr\} + O\Bigl( \frac1N \Bigr).
\ee
Notice that $O(c^{4})$ is bounded uniformly in $N$, and $O(\frac1N)$ is bounded uniformly in $c$. They depend on $|\Lambda|$ and $\beta$, on the other hand, but it does not matter. Let $A \subset \Omega$ be the event where a double bar occurs on the edge $\{0,e_{1}\}$ in the time interval $[0,\frac\beta N]$. Expanding the exponential, and using translation invariance in space and time, we get
\be
Z_{N}(v^{*}) = Z_{N}(0) \bigl[ 1 - \tfrac{4bc^{2}}\beta |\Lambda| \bigr] + \tfrac{8a^{2}c^{2}}{\beta^{2}} |\caE| N \int_{A} \dd\rho_{u}(\omega) \sum_{\sigma\in\Sigma(\omega)} \bigl( \vec\sigma_{0,\frac\beta N}^{(1)} - \vec\sigma_{0,0}^{(1)} \bigr)^{2} + O(c^{4}) + O(\tfrac1N).
\ee
We now estimate the contribution of the realizations with an enforced double bar on $(0,e_{1})$ at a time close to 0. As before, we denote $\bar\omega$ the realization $\omega$ with an extra double bar at $\{0,e_{1}\} \times 0$.
\be
\begin{split}
\lim_{N\to\infty} \frac1{Z_{N}(0)} &\frac N\beta \int_{A} \dd\rho_{u}(\omega) \sum_{\sigma\in\Sigma(\omega)} \bigl( \vec\sigma_{0,\frac\beta N}^{(1)} - \vec\sigma_{0,0}^{(1)} \bigr)^{2} \\
&= (1-u) \lim_{N\to\infty} \frac1{Z_{N}(0)} \int_{E_{0,e_{1},0}} \dd\rho_{u}(\omega) 1_{E_{0,0,0+}^{\rm c}}(\bar\omega) \sum_{\sigma\in\Sigma(\bar\omega)} \bigl( \vec\sigma_{0,\frac\beta N}^{(1)} - \vec\sigma_{0,0}^{(1)} \bigr)^{2} \\
&= (1-u) \bbP(E_{0,e_{1},0}^{-}) \\
&\leq (1-u) \bbP(E_{0,e_{1},0}).
\end{split}
\ee
The sum over space-time spin configurations that are compatible with $\omega \in E_{0,e_{1},0}$, was computed as follows:
\be
\begin{split}
\sum_{\sigma\in\Sigma(\bar\omega)} \bigl( \vec\sigma_{0,\frac\beta N}^{(1)} - \vec\sigma_{0,0}^{(1)} \bigr)^{2} &= (2S+1)^{|\caL(\omega)|-1} \sum_{a,b=-S}^{S} \bigl( \vec a^{(1)} - \vec b^{(1)} \bigr)^{2} \\
&= (2S+1)^{|\caL(\omega)|-1} \frac1{2S} \sum_{i=1}^{2S} \sum_{a,b=-S}^{S} \bigl( \vec a^{(i)} - \vec b^{(i)} \bigr)^{2} \\
&= (2S+1)^{|\caL(\omega)|-1} \frac1{2S} \sum_{a,b=-S}^{S} \| \vec a - \vec b \|^{2} \\
&= (2S+1)^{|\caL(\omega)|}.
\end{split}
\ee
We used Eqs \eqref{simplex}. We have obtained
\be
Z_{N}(v^{*}) \leq Z_{N}(0) \Bigl[ 1 - \tfrac{4bc^{2}}\beta |\Lambda| + \tfrac{8a^{2}c^{2}}\beta |\caE| (1-u) \kappa(e_{1},0) \Bigr] + O(c^{4}) + O(\tfrac1N).
\ee
We see that $c=0$ is local maximizer whenever $\frac{4b}\beta > \frac{8a^{2}d}\beta (1-u) \kappa(e_{1},0)$, which yields the relation between $a$ and $b$ that is stated in the proposition.
\end{proof}

We now use Proposition \ref{prop gen GD} in order to get a generalized infrared bound for $\widetilde\kappa(k,\tau)$. This is similar to Proposition \ref{prop irb Duhamel}.

\begin{proposition}
\label{prop gen irb}
Assume that there exists $c_{0}>0$ such that the partition function in Eq. \eqref{fpart gen fields} satisfies $Z(v) \leq Z(0)$ for every $v \in \caV_{c_{0}}$. Then for all $(k,\tau) \neq (0,0)$, we have
\[
\widetilde\kappa(k,\tau) \leq (2S+1) \frac{\eps(k) + 4b\tau^{2}}{(\eps(k) + a\tau^{2})^{2}}.
\]
\end{proposition}

\begin{proof}
We choose $v_{xt} = \cos(kx+\tau t)$. We have
\be
\begin{split}
Z(\eta v) = &\int\dd\rho_{u}(\omega) \sum_{\sigma\in\Sigma(\omega)} \exp\biggl\{ \int_{0}^{\beta} \dd t \Bigl[ \eta \bigl( \vec\sigma_{\cdot,t}^{(1)}, \Delta v_{\cdot,t} \bigr) + \tfrac14 \eta^{2} (v_{\cdot,t}, \Delta v_{\cdot,t}) \Bigr] \\
&+ \sum_{x\in\Lambda} \int_{0}^{\beta} \dd t \Bigl[ a\eta \vec\sigma_{xt}^{(1)} \frac{\partial^{2} v_{xt}}{\partial t^{2}} + b\eta^{2}
 v_{xt} \frac{\partial^{2} v_{xt}}{\partial t^{2}} \Bigr] \biggr\}.
\end{split}
\ee
We now use $-\Delta v = \eps(k) v$ and $-\frac{\partial^{2}}{\partial t^{2}} v = \tau^{2} v$, and we get
\be
\begin{split}
Z(\eta v) &= \int\dd\rho_{u}(\omega) \sum_{\sigma\in\Sigma(\omega)} \exp\biggl\{ -\int_{0}^{\beta} \dd t \Bigl[ \eta \bigl( \eps(k) + a\tau^{2} \bigr) \bigl( \vec\sigma_{\cdot,t}^{(1)}, v_{\cdot,t} \bigr) \\
&\hspace{70mm} + \tfrac14 \eta^{2} \bigl( \eps(k) + 4b \tau^{2} \bigr) (v_{\cdot,t}, v_{\cdot,t}) \Bigr] \biggr\} \\
&= \int\dd\rho_{u}(\omega) \sum_{\sigma\in\Sigma(\omega)} \biggl( 1 + \tfrac12 \eta^{2} \bigl( \eps(k) + a\tau^{2} \bigr)^{2} \int_{0}^{\beta} \dd t \int_{0}^{\beta} \dd t' \bigl( \vec\sigma_{\cdot,t}^{(1)}, v_{\cdot,t} \bigr) \bigl( \vec\sigma_{\cdot,t'}^{(1)}, v_{\cdot,t'} \bigr) \\
&\hspace{5cm} - \tfrac14 \eta^{2} \bigl( \eps(k) + 4b \tau^{2} \bigr) \int_{0}^{\beta} \dd t (v_{\cdot,t}, v_{\cdot,t}) + O(\eta^{4}) \biggr).
\end{split}
\ee

We have
\be
\bigl( \vec\sigma_{\cdot,t}^{(1)}, v_{\cdot,t} \bigr) \bigl( \vec\sigma_{\cdot,t'}^{(1)}, v_{\cdot,t'} \bigr) = \sum_{x,y\in\Lambda} \cos(kx+\tau t) \cos(ky+\tau t') \vec\sigma_{xt}^{(1)} \vec\sigma_{yt'}^{(1)},
\ee
and, using the symmetries of the cubic box and Eqs \eqref{simplex},
\be
\begin{split}
\frac1{Z(0)} \int\dd\rho_{u}(\omega) \sum_{\sigma\in\Sigma(\omega)} \vec\sigma_{xt}^{(1)} \vec\sigma_{yt'}^{(1)} &= \frac1{Z(0)} \int_{E_{0,y-x,t-t'}} \dd\rho_{u}(\omega) (2S+1)^{|\caL(\omega)|-1} \frac1{2S} \sum_{a=-S}^{S} \|\vec a\|^{2} \\
&= \frac1{2(2S+1)} \kappa(y-x,t'-t).
\end{split}
\ee
Finally, we obtain the Fourier transform of correlation functions:
\be
\begin{split}
\int_{0}^{\beta} \dd t \int_{0}^{\beta} \dd t' &\sum_{x,y \in \Lambda} \cos(kx+\tau t) \cos(ky+\tau t') \, \kappa(y-x,t'-t) \\
&= \int_{0}^{\beta} \dd t \int_{0}^{\beta} \dd t'' \sum_{x,z \in \Lambda} \cos(kx+\tau t) \cos(k(x+z) + \tau(t+t'')) \, \kappa(z,t'') \\
&= \int_{0}^{\beta} \sum_{x\in\Lambda} \cos(kx+\tau t) \Re \e{\ii kx + \ii \tau t} \int_{0}^{\beta} \dd t'' \e{\ii kz + \ii \tau t''} \kappa(z,t'') \\
&= \int_{0}^{\beta} \dd t \sum_{x\in\Lambda} \cos(kx+\tau t)^{2} \widetilde\kappa(-k,-\tau).
\end{split}
\ee
We actually have $\widetilde\kappa(-k,-\tau) = \widetilde\kappa(k,\tau)$ because of symmetries. We have obtained
\bm
Z(\eta v) = Z(0) \biggl[ 1 + \frac1{4(2S+1)} \eta^{2} \bigl( \eps(k) + a\tau^{2} \bigr)^{2} \widetilde\kappa(k,\tau) \int_{0} \int_{0}^{\beta} \dd t (v_{\cdot,t}, v_{\cdot,t}) \\
- \tfrac14 \eta^{2} \bigl( \eps(k) + 4b\tau^{2} \bigr)^{2} \int_{0}^{\beta} \dd t (v_{\cdot,t}, v_{\cdot,t}) + O(\eta^{4}) \biggr].
\end{multline}
The conclusion follows.
\end{proof}

\begin{corollary}
Assume $u \in [0,\frac12]$. Then for all $(k,\tau) \neq (0,0)$, we have
\[
\widetilde\kappa(k,\tau) \leq \frac{2S+1}{\eps(k) + \frac{\tau^{2}}{8d (1-u) \kappa(e_{1},0)}}.
\]
\end{corollary}

\begin{proof}
It follows from Proposition \ref{prop gen GD} that the bound of Proposition \ref{prop gen irb} holds with $b = 2d a^{2} (1-u) \kappa(e_{1},0)$, for any $a>0$. We get the result by optimizing over $a$.
\end{proof}

\begin{corollary}
\label{cor gen irb}
Assume $u \in [0,\frac12]$. Then for all $k \neq 0$, we have
\[
\widehat\kappa(k,0) = \frac1\beta \sum_{\tau \in \frac{2\pi}\beta \bbZ} \widetilde\kappa(k,\tau) \leq \frac{(2S+1) \sqrt{2d (1-u) \kappa(e_{1},0)}}{\sqrt{\eps(k)}} \coth \Bigl( \beta \sqrt{2d (1-u) \kappa(e_{1},0) \eps(k)} \Bigr).
\]
\end{corollary}

This corollary follows from the identity $\sum_{n\in\bbZ} \frac1{n^{2}+\xi^{2}} = \frac\pi\xi \coth(\pi\xi)$. We can now prove Theorem \ref{thm gen irb}.

\begin{proof}[Proof of Theorem \ref{thm gen irb}]
For the first bound, we use
\be
\begin{split}
1 = \kappa(0,0) &= \frac1{|\Lambda| \beta} \sum_{k\in\Lambda^{*}} \sum_{\tau \in \frac{2\pi}\beta \bbZ} \widetilde\kappa(k,\tau) \\
&= \frac1{|\Lambda| \beta} \widetilde\kappa(0,0)  + \frac1{|\Lambda| \beta} \sum_{\tau \in \frac{2\pi}\beta \bbZ \setminus \{0\}} \widetilde\kappa(0,\tau)  + \frac1{|\Lambda| \beta} \sum_{k\in\Lambda^{*} \setminus \{0\}} \sum_{\tau \in \frac{2\pi}\beta \bbZ} \widetilde\kappa(k,\tau).
\end{split}
\ee
The first term is equal to $\bbE \bigl( \frac{L_{(0,0)}}{\beta|\Lambda|} \bigr)$. The middle term vanishes in the limit $|\Lambda|\to\infty$. The last term can be bounded by Corollary \ref{cor gen irb}, recalling that $\widetilde\kappa$ is real because of lattice symmetries. We get
\bm
\lim_{|\Lambda|\to\infty} \bbE \Bigl( \frac{L_{(0,0)}}{\beta|\Lambda|} \Bigr) \geq 1 - (2S+1) \sqrt{2d (1-u) \kappa(e_{1},0)} \\
\times \frac1{(2\pi)^{d}} \int_{[-\pi,\pi]^{d}} \frac{\coth(\beta \sqrt{2d (1-u) \kappa(e_{1},0) \eps(k)})}{\sqrt{\eps(k)}} \dd k.
\end{multline}
The cotangent disappears in the limit $\beta\to\infty$ by dominated convergence (in $d\geq3$) and we obtain the first bound of Theorem \ref{thm gen irb} since the integral of $1/\sqrt{\eps(k)}$ is equal to $J_{d}' / \sqrt{2d}$.

For the second bound we use the sum rule of \cite{KLS1}. Using invariance under lattice rotations, we have
\be
\begin{split}
\kappa(e_{1},0) &= \frac1{d |\Lambda| \beta} \sum_{k\in\Lambda^{*}} \sum_{\tau \in \frac{2\pi}\beta \bbZ} \widetilde\kappa(k,\tau) \sum_{i=1}^{d} \cos k_{i} \\
&= \frac1{|\Lambda| \beta} \widetilde\kappa(0,0)  + \frac1{|\Lambda| \beta} \sum_{\tau \in \frac{2\pi}\beta \bbZ \setminus \{0\}} \widetilde\kappa(0,\tau)  + \frac1{d|\Lambda| \beta} \sum_{k\in\Lambda^{*} \setminus \{0\}} \sum_{\tau \in \frac{2\pi}\beta \bbZ} \widetilde\kappa(k,\tau) \sum_{i=1}^{d} \cos k_{i}.
\end{split}
\ee
As before, the first term is equal to $\bbE \bigl( \frac{L_{(0,0)}}{\beta|\Lambda|} \bigr)$ and the middle term vanishes in the limit $|\Lambda|\to\infty$. Using Corollary \ref{cor gen irb}, we get
\bm
\lim_{|\Lambda|\to\infty} \bbE\Bigl( \frac{L_{(0,0)}}{\beta|\Lambda|} \Bigr) \geq \kappa(e_{1},0) - (2S+1) \sqrt{2d (1-u) \kappa(e_{1},0)} \\
\times \frac1{d(2\pi)^{d}} \int_{[-\pi,\pi]^{d}} \frac{\coth(\beta \sqrt{2d (1-u) \kappa(e_{1},0) \eps(k)})}{\sqrt{\eps(k)}} \Bigl( \sum_{i=1}^{d} \cos k_{i} \Bigr)_{+} \dd k.
\end{multline}
The cotangent again disappears in the limit $\beta\to\infty$ and the integral is equal to $I_{d}' / \sqrt{2d}$.
\end{proof}

\section{Specific models of interest}
\label{sec models}

It is time to give flesh to the quantum models under study. The Hamiltonians were defined in terms of operators $T_{xy}$, $Q_{xy}$, and $P_{xy}$, that were chosen for their mathematical convenience rather than their physical relevance. In this section we discuss the special cases $S=\frac12$ and $S=1$ in details.

\subsection{Spin $\frac12 $ models}
\label{sec spin12}

Let us start with the {\bf spin $\frac12$ Heisenberg ferromagnet}, the model that was studied by Conlon and Solovej \cite{CS}, and by T\'oth who introduced the representation with ``crosses'' \cite{Toth1}. The parameters are $S=\frac12$ and $u=1$. The loop parameter is thus $\theta=2$. Using $\vec S_{x} \cdot \vec S_{y} = \frac12 T_{xy} - \frac14 \Id$, we find that
\be
H_{\Lambda}^{(1)} = -2 \sum_{\{x,y\} \in \caE} \bigl( \vec S_{x} \cdot \vec S_{y} - \tfrac14 \bigr).
\ee
Since only transpositions are present, the loop representation can be seen as a model of ``spatial permutations'', i.e., bijections $\Lambda \to \Lambda$. Each loop corresponds to a permutation cycle, and the length of the loop is equal to $\beta$ times the number of vertices in the cycle.

Next, we discuss the {\bf spin $\frac12$ Heisenberg antiferromagnet}.
Let $S=\frac12$ and $u=0$, and consider the Hamiltonian $\tilde H_{\Lambda}^{(0)} = -\sum_{\{x,y\}} P_{xy}$. We have $\vec S_{x} \cdot \vec S_{y} = \frac14 \Id - \frac12 P_{xy}$, so that
\be
\tilde H_{\Lambda}^{(0)} = 2 \sum_{\{x,y\} \in \caE} \bigl( \vec S_{x} \cdot \vec S_{y} + \tfrac14 \bigr).
\ee
This is indeed the Hamiltonian of the Heisenberg antiferromagnet. We cannot use Theorem \ref{thm integer spin} because the spin is half-integer. But if we assume in addition that the graph is bipartite, i.e., $\Lambda = \Lambda_{\rm A} \cup \Lambda_{\rm B}$ such that all edges of $\caE$ involve one site in $\Lambda_{\rm A}$ and one site in $\Lambda_{\rm B}$, the Hamiltonian is unitarily equivalent to $H_{\Lambda}^{(0)}$:
\be
H_{\Lambda}^{(0)} = \Bigl( \prod_{x \in \Lambda_{\rm B}} \e{-\ii \pi S_{x}^{2}} \Bigr) \tilde H_{\Lambda}^{(0)} \Bigl( \prod_{x \in \Lambda_{\rm B}} \e{\ii \pi S_{x}^{2}} \Bigr).
\ee
Then we can use the probabilistic representation. It only involves double bars since $u=0$, and it was introduced by Aizenman and Nachtergaele \cite{AN}. Spin correlations are given by
\be
\langle S_{x}^{i} S_{y}^{i} \rangle = \tfrac14 (-1)^{x-y} \bbP(E_{x,y,0})
\ee
for $i=1,2,3$, where $(-1)^{x-y}$ is equal to 1 if $x,y$ belong to the same sublattice, $-1$ otherwise.

The case of {\it frustrated} systems where the graph is not bipartite is currently attracting a lot of attention by condensed matter physicists. The probabilistic representation does not apply, because the weights of loops would carry signs.

We can check that
\be
Q_{xy} = 2 \bigl( S_{x}^{1} S_{y}^{1} - S_{x}^{2} S_{y}^{2} + S_{x}^{3} S_{y}^{3} \bigr) + \tfrac12.
\ee
We then obtain a family of Heisenberg models with anisotropic spin interactions, namely
\be
H_{\Lambda}^{(u)} = -2 \sum_{\{x,y\}\in\caE} \bigl( S_{x}^{1} S_{y}^{1} + (2u-1) S_{x}^{2} S_{y}^{2} + S_{x}^{3} S_{y}^{3} - \tfrac14 \bigr).
\ee
The case $u=\frac12$ gives the {\bf spin $\frac12$ XY model}. 
A consequence of Theorem \ref{thm corr} is that
\be
0 \leq | \langle S_{x}^{2} S_{y}^{2} \rangle | \leq \langle S_{x}^{1} S_{y}^{1} \rangle = \langle S_{x}^{3} S_{y}^{3} \rangle.
\ee
Neither the second inequality, nor the positivity of the latter correlations, seem to be immediate.

An additional motivation for the XY model comes from the fact that it is equivalent to the {\bf hard-core Bose gas}. This is well-known, see e.g.\ \cite{KLS1,ALSSY}, but we recall it for the convenience of the readers. Notice that the present XY model differs somewhat from standard conventions, since interactions are between spins in the directions 1 and 3 rather than 1 and 2, and this requires a few modifications of the usual correspondence.

Let $a_{x} = S_{x}^{1} + \ii S_{x}^{3}$ and its adjoint $a_{x}^{\dagger} = S_{x}^{1} - \ii S_{x}^{3}$. These operators satisfy the commutation relations
\be
\begin{split}
&[a_{x}, a_{y}^{(\dagger)}] = 0 \quad \text{for all } x \neq y, \\
&\{ a_{x}, a_{x}^{\dagger} \} = \Id_{\Lambda} \quad \text{for all } x \in \Lambda.
\end{split}
\ee
In addition, the operator for the number of particles at $x$ is
\be
n_{x} = a_{x}^{\dagger} a_{x} = S_{x}^{2} + \tfrac12 \Id_{\Lambda}.
\ee
The Hamiltonian can be rewritten as
\be
H_{\Lambda}^{(\frac12)}  = -\sum_{\{x,y\} \in \caE} \Bigl( a_{x}^{\dagger} a_{y} + a_{y}^{\dagger} a_{x} - \tfrac12 \Bigr).
\ee

The relevant correlation functions are those representing ``off-diagonal long-range order'', that signal the occurrence of Bose-Einstein condensation:
\be
\begin{split}
\langle a_{x}^{\dagger} a_{y} \rangle &= \frac1{Z_{\Lambda}^{(\frac12)}} \Tr (S_{x}^{1} - \ii S_{x}^{3}) (S_{y}^{1} + \ii S_{y}^{3}) \e{-\beta H_{\Lambda}^{(\frac12)}} \\
&= 2 \langle S_{x}^{3} S_{y}^{3} \rangle.
\end{split}
\ee
We used the identity $\langle S_{x}^{1} S_{y}^{1} \rangle = \langle S_{x}^{3} S_{y}^{3} \rangle$, and the fact that $\langle S_{x}^{3} S_{y}^{1} \rangle = -\langle S_{x}^{1} S_{y}^{3} \rangle = 0$. The density-density correlation function, on the other hand, is given by
\be
\langle n_{x} n_{y} \rangle - \langle n_{x} \rangle \langle n_{y} \rangle = \langle S_{x}^{2} S_{y}^{2} \rangle.
\ee
The latter correlation function is given by difference of probabilities of $E_{x,y,0}^{+}$ and $E_{x,y,0}^{-}$, see Theorem \ref{thm corr}; we conjecture below, in Conjecture 1, that it has exponential decay with respect to $\|x-y\|$. This is indeed expected in the case of the hard-core Bose gas.

Existence of long-range order was proved in \cite{DLS,KLS1} in rectangular boxes, $d\geq3$, and low temperatures. Theorem \ref{thm irb} does not improve these results. But we use the random loop representation to give a heuristic description of phase transitions and symmetry breaking in Section \ref{sec outlook}.

\subsection{Spin 1 SU(2)-invariant model}
\label{sec spin1}

It is well-known that any two-body SU(2)-invariant interaction for $S=1$ can be be written as
\be
\label{general spin 1 Hamiltonian}
\check H_{\Lambda} = -\sum_{\{x,y\}\in\caE} \Bigl( J_{1} \vec S_{x} \cdot \vec S_{y} + J_{2} (\vec S_{x} \cdot \vec S_{y})^{2} \Bigr).
\ee

It may be worth discussing first the ground state phase diagram of the classical model, where each site is the host of a vector in $\bbS^{2}$. There is a ferromagnetic phase when $J_{1}>0, J_{2}\geq0$, an antiferromagnetic phase when $J_{1}<0, J_{2}\geq0$, and a nematic phase when $J_{1}=0, J_{2}>0$. Results for low temperatures have been obtained in \cite{FSS} in the case of the classical Heisenberg models ($J_{2}=0$), and in \cite{AZ,BC} in the case of the classical nematic model ($J_{1}=0, J_{2}>0$). The case $J_{2}<0$ would be interesting to scrutinize.

The phase diagram of the quantum model has been investigated by several authors, see \cite{BO,TZX,TLMP,FKK} and references therein, and it differs significantly from the classical one. It is displayed in Fig.\ \ref{fig phd}. The line $J_{2}=0$ corresponds to the usual Heisenberg models. The line in the direction $(-1,-\frac13)$ corresponds to the model introduced by Affleck, Kennedy, Lieb, and Tasaki in order to study Haldane's conjecture \cite{AKLT}.

\bfig
\includegraphics[width=10cm]{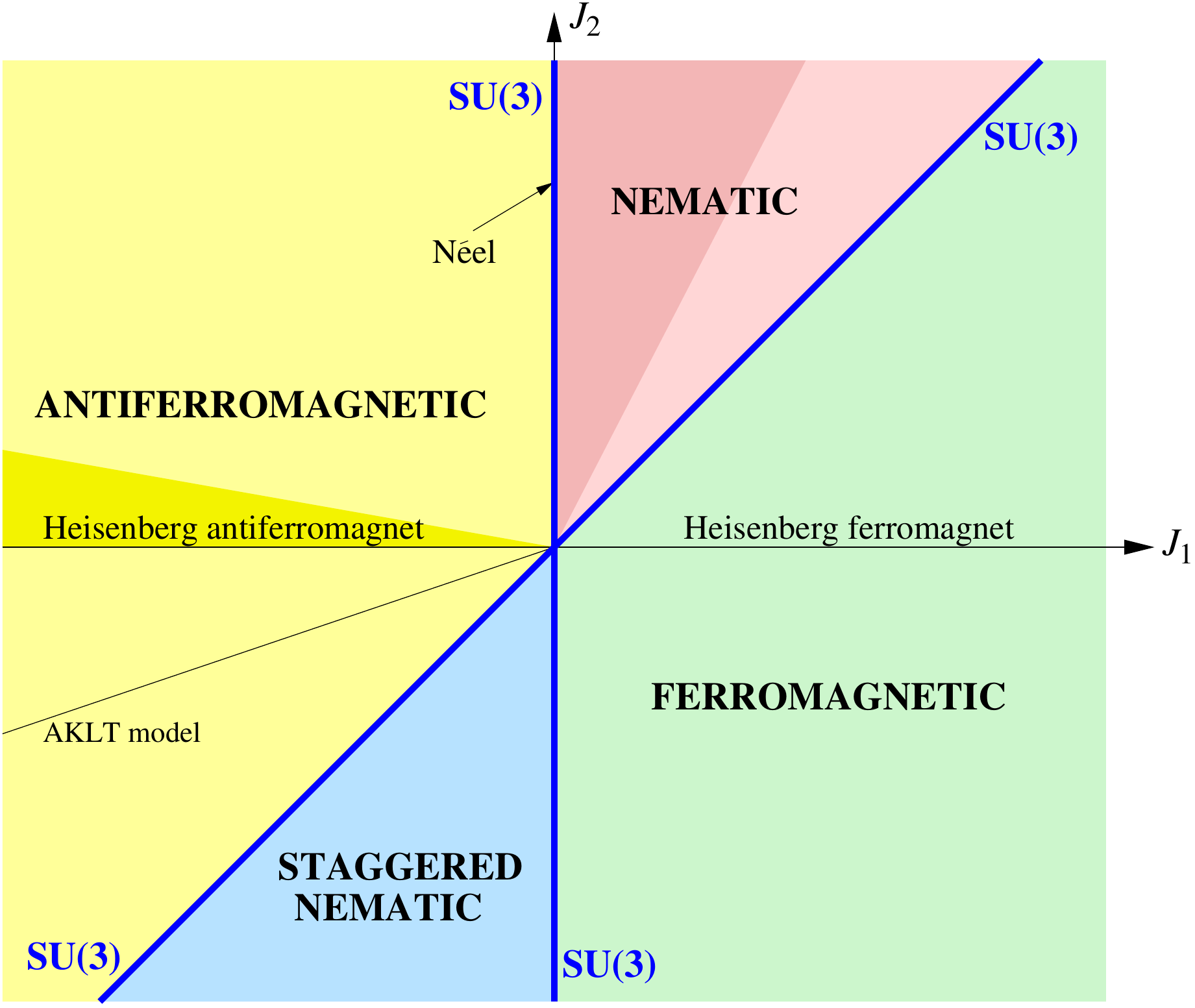}
\caption{(Color online) Phase diagram of the general spin 1 model with Hamiltonian \eqref{general spin 1 Hamiltonian} in dimension $d\geq3$. On the two lines $J_{1}=0$ and $J_{2}=J_{1}$ the model has SU(3) invariance, not only SU(2). The phase diagram is expected to show four phases (ferromagnetic, nematic, antiferromagnetic, staggered nematic) that are separated by those lines. Antiferromagnetic long-range order has been proved in the dark yellow region \cite{DLS,KLS1}. The random loop representation allows to prove N\'eel order for $J_{1}=0$ and $J_{2}>0$ (Theorem \ref{thm spin 1 af}), and to prove another form of magnetic order in the dark pink region $0 < J_{1} \leq \frac12 J_{2}$, that is compatible with a nematic phase (Theorem \ref{thm spin 1}).}
\label{fig phd}
\efig

It is possible to check that $J_{1} \vec S_{x} \cdot \vec S_{y}$ is reflection positive, in the quantum sense, when $J_{1} \leq 0$ and that $J_{2} (\vec S_{x} \cdot \vec S_{y})^{2}$ is reflection positive when $J_{2} \geq 0$. Thus $\check H_{\Lambda}$ is definitely reflection positive in the quadrant $J_{1} \leq 0, J_{2} \geq 2$. Long-range order has been proved for the antiferromagnet when $d\geq3$ and when the temperature is low enough \cite{DLS,KLS1}. One can obtain an infrared bound for the usual correlation function, which allows to extend the domain of long-range order to the dark yellow domain depicted in Fig.\ \ref{fig phd}.
The domain of reflection positivity presumably extends a bit beyond the quadrant, but this has not been shown yet.

The Hamiltonian $\tilde H^{(u)}_{\Lambda}$ defined in Eq.\ \eqref{def fam2} is SU(2) invariant, and is therefore of the form \eqref{general spin 1 Hamiltonian} up to a shift by the identity operator. Let us express $T_{xy}$ and $P_{xy}$ as linear combinations of $\vec S_{x} \cdot \vec S_{y}$ and $(\vec S_{x} \cdot \vec S_{y})^{2}$.

\begin{lemma}
\label{lem spin 1}
In the case $S=1$, we have
\[
\begin{split}
&T_{xy} = \vec S_{x} \cdot \vec S_{y} + (\vec S_{x} \cdot \vec S_{y})^{2} - 1, \\
&P_{xy} = (\vec S_{x} \cdot \vec S_{y})^{2} - 1.
\end{split}
\]
\end{lemma}

\begin{proof}
Let $S_{x}^{\pm} = S_{x}^{1} \pm \ii S_{x}^{2}$. It is well-known that, in the basis where $S_{x}^{3}$ is diagonal, we have
\be
S^{\pm}_{x} |a\rangle = \sqrt{2 - a(a\pm1)} |a\pm1\rangle,
\ee
with the understanding that $S^{+}_{x} |1\rangle = 0$ and $S^{-}_{x} |-1\rangle = 0$. Using the identity
\be
2 \vec S_{x} \cdot \vec S_{y} = S_{x}^{+} S_{y}^{-} + S_{x}^{-} S_{y}^{+} + 2 S_{x}^{3} S_{y}^{3},
\ee
the claim for $T_{xy}$ can be verified by direct calculations of all matrix elements.

Direct calculations can also be used for $P_{xy}$. However, a more elegant argument uses properties of additions of spins (see e.g.\ \cite{Mes}). It is well-known that the eigenvalues of $(\vec S_{x} + \vec S_{y})^{2}$ are 0, 2, 6 (they are equal to $J(J+1)$ with $J \in \{0,1,2\}$). This gives the eigenvalues for $\vec S_{x} \cdot \vec S_{y}$: $-2$, $-1$, 1, and hence for $(\vec S_{x} \cdot \vec S_{y})^{2}$: 4, 1, 1. Recall that $\frac13 P_{xy}$ is the projector onto the one-dimensional eigenspace of $(\vec S_{x} + \vec S_{y})^{2}$ with eigenvalue 0. The identity of the lemma is now easily checked for any vector that belongs to the eigensubspaces.
\end{proof}

It follows from Lemma \ref{lem spin 1} that $\tilde H_{\Lambda}^{(u)}$ can be written as
\be
\tilde H_{\Lambda}^{(u)} = -\sum_{\{x,y\}\in\caE} \Bigl( u \vec S_{x} \cdot \vec S_{y} + (\vec S_{x} \cdot \vec S_{y})^{2} - 2 \Bigr).
\ee

The region of parameters where the model has the probabilistic representation (with positive weights of the loops) is delimited by the lines $J_{1}=0$ and $J_{2}=J_{1}$. This is precisely the pink region of the nematic phase. The ordinary spin-spin correlation function is given by Theorem \ref{thm integer spin} (a), and this leads to the following conjecture, that is indeed compatible with a nematic phase.

\begin{conjecture}
\label{conj}
Let $(J_{1},J_{2})$ satisfy $0 < J_{1} < J_{2}$. For all $\beta\in\bbR$ and all $d\geq1$, the correlation function
\[
\langle S_{x}^{i} S_{y}^{i} \rangle = \tfrac23 \bigl[ \bbP(E_{x,y,t}^{+}) - \bbE(E_{x,y,t}^{-}) \bigr]
\]
has exponential decay with respect to $\|x-y\|$.
\end{conjecture}

Long-range correlations are only possible if long loops are present, and the vertical orientation is quickly lost. It should be possible to prove this rigorously, although it does not appear to be straightforward when the dimension $d$ is larger than 1.

The results of Section \ref{sec macroscopic loops} nonetheless imply a phase transition with long-range order. It follows from Theorem \ref{thm integer spin} (b) that macroscopic loops are accompanied by long-range correlations of $(S_{x}^{i})^{2}$. Here is the main result of this article for the model with $S=1$.

\begin{theorem}
\label{thm spin 1}
Let $(\Lambda,\caE)$ be the $d$-dimensional cubic box with periodic boundary conditions and even side length. Assume that $0 \leq J_{1} \leq \frac12 J_{2}$ and that $d\geq5$. Then there exists $\beta_{0} < \infty$ and $c>0$ such that
\[
\frac1{|\Lambda|^{2}} \sum_{x,y\in\Lambda} \Bigl( \langle (S_{x}^{i})^{2} (S_{y}^{i})^{2} \rangle - \langle (S_{x}^{i})^{2} \rangle \langle (S_{y}^{i})^{2} \rangle \Bigr) \geq c
\]
for all $\beta > \beta_{0}$, $i=1,2,3$, $x,y \in \Lambda$, uniformly in the size of the system.
\end{theorem}

This theorem establishes the existence of a phase transition with symmetry breaking, since uniqueness of infinite-volume Gibbs states implies the decay of all correlations. Such magnetic order is compatible with the nematic phase. Theorem \ref{thm spin 1} is a direct consequence of Theorem \ref{thm integer spin} (b) and Theorem \ref{thm irb}. It actually holds for $d\geq3$ when $J_{1} \lesssim \frac12 J_{2}$.

The case $u=0$ in a bipartite graph is different. Only double bars occur in the loop representation, and the vertical orientation displays alternating properties. Namely, $\bbP(E_{x,y,t}^{-}) = 0$ whenever $x,y$ belong to the same sublattice, and $\bbP(E_{x,y,t}^{+}) = 0$ whenever $x,y$ belong to different sublattices. We therefore obtain the existence of N\'eel order at low temperatures, which stands in stark contrast to the classical case.

\begin{theorem}
\label{thm spin 1 af}
Let $(\Lambda,\caE)$ be the $d$-dimensional cubic box with periodic boundary conditions and even side length. Assume that $J_{1} = 0$, $J_{2} > 0$, and that $d\geq5$. Then there exist $\beta_{0} < \infty$ and $c>0$ such that
\[
\frac1{|\Lambda|^{2}} \sum_{x,y\in\Lambda} (-1)^{x-y} \langle S_{x}^{i} S_{y}^{i} \rangle \geq c
\]
for all $\beta > \beta_{0}$, $i=1,2,3$, $x,y \in \Lambda$, uniformly in the size of the system.
\end{theorem}

This theorem also follows directly from Theorem \ref{thm irb}. N\'eel order certainly occurs in dimensions $d=3,4$ as well.

\section{Conclusion and outlook}

Connections between random loop models and quantum lattice systems provide many deep insights for each of them. The continuous symmetries of the spin systems have far-reaching consequences regarding the size of loops in dimensions 1 and 2. Results about long-range order, obtained in the quantum setting \cite{DLS,KLS1}, establish the occurrence of macroscopic loops when $S=\frac12$, i.e., $\theta=2$. We have extended the result to higher values of $S$ and $\theta$. We proved this by adapting the method of reflection positivity and infrared bounds of \cite{FSS}, rearranging the underlying Poisson point process of transitions so that it becomes reflection positive. The case $S=1$ turns out to correspond to the nematic phase of a very interesting SU(2)-invariant quantum system.

Several of the present results should hold more generally. Long loops should be absent in dimension  2 for all $\theta>0$, not only $\theta=2,3,\dots$ Occurrence of macroscopic loops is proved for small $S$ or large $d$ (and large $\beta$); our conditions could certainly be loosened. The restriction to $u \leq \frac12$ seems to be an inherent feature of the method; it is indeed present in the quantum setting, having frustrated experts since the method was introduced.

Nachtergaele's extension of the loop representation for higher spins \cite{Nac1,Nac2} is different from the ones discussed here. Comparing the information provided by both could lead to new results.

\subsection{Joint distribution of the lengths of macroscopic loops}

This is an intriguing topic for future research, both in itself and for the information it provides on the structure of pure Gibbs states and symmetry breaking. It seems appropriate to discuss this in details. Let $L^{(1)}, L^{(2)}, \dots$ denote the lengths of the loops in decreasing order. Clearly, $\bigl( \frac{L^{(1)}}{\beta |\Lambda|}, \frac{L^{(2)}}{\beta |\Lambda|}, \dots \bigr)$ is a random partition of $[0,1]$. In the case of the random interchange model on the complete graph, i.e., $\theta=1$, $u=1$, and $(\Lambda,\caE)$ is the complete graph, Aldous conjectured that this random partition has Poisson-Dirichlet(1) distribution. This was subsequently proved by Schramm \cite{Sch}, who showed that the time evolution of the loop lengths is described by an effective split-merge process (or ``coagulation-fragmentation''). The relevance of these ideas in the presence of spatial structure (that is, $\Lambda \subset \bbZ^{d}$) and for $\theta=2$ was explained in \cite{GUW}. This is backed by results for the model of ``spatial random permutations'', that are rigorous in the annealed case \cite{BU} and numerical in the quenched lattice case \cite{GLU}.

We start by describing the heuristics. Our main goal is to justify Conjecture \ref{conj joint distr} below. We assume that the graph is regular, $\Lambda \subset \bbZ^{d}$, but the discussion holds more generally.

Macroscopic loops are spread all over $\Lambda$ and they interact between one another, and among themselves, in an essentially mean-field fashion. We introduce a Markov process for which the measure \eqref{La mais-je} is invariant. It is enough that it satisfies the detailed balance property, which can be written as follows:
\bm
\theta^{|\caL(\omega)|} (u \dd t)^{\# \text{ crosses in $\omega$}} \bigl( (1-u) \dd t \bigr)^{\# \text{ double bars in $\omega$}} R(\omega,\omega') \\
= \theta^{|\caL(\omega')|} (u \dd t)^{\# \text{ crosses in $\omega'$}} \bigl( (1-u) \dd t \bigr)^{\# \text{ double bars in $\omega'$}} R(\omega',\omega).
\end{multline}
We have discretized the interval $[0,\beta]$ with mesh $\dd t$, and $R(\omega,\omega')$ denotes the rate at which $\omega'$ occurs when the configuration is $\omega$. The following process satisfies the detailed balance property.
\begin{itemize}
\item A new cross appears in the interval $\{x,y\} \times [t,\dd t]$ at rate $u \theta^{1/2} \dd t$ if its appearance causes a loop to split; at rate $u \theta^{-1/2} \dd t$ if its appearance causes two loops to merge; and at rate $u \dd t$ if its appearance does not modify the number of loops.
\item Same with double bars, but with $1-u$ instead of $u$.
\item An existing cross and double bar is removed at rate $\theta^{1/2}$ if its removal causes a loop to split; at rate $\theta^{-1/2}$ if its removal causes two loops to merge; and at rate 1 if the number of loops remains constant.
\end{itemize}
Notice that any new cross or double bar between two loops causes them to merge. When $u=1$, any new cross within a loop causes it to split.

Let $\gamma, \gamma'$ be two long loops of respective lengths $L,L'$. A new cross or double bar that causes $\gamma$ to split appears at rate $\tfrac12 c_{1} \theta^{1/2} \frac{L^{2}}{\beta |\Lambda|}$; a new cross or double bar that causes $\gamma$ and $\gamma'$ to merge appears at rate $c_{1} \theta^{-1/2} \frac{L L'}{\beta |\Lambda|}$. The rate for an existing cross or double bar to disappear is $\tfrac12 c_{2} \theta^{1/2} \frac{L^{2}}{\beta |\Lambda|}$ if $\gamma$ is split, and $c_{2} \theta^{-1/2} \frac{L L'}{\beta |\Lambda|}$ if $\gamma$ and $\gamma'$ are merged. Consequently, $\gamma$ splits at rate
\be
\tfrac12 (c_{1}+c_{2}) \theta^{1/2} \frac{L^{2}}{\beta |\Lambda|} \equiv \tfrac12 r_{\rm s} L^{2}
\ee
and $\gamma, \gamma'$ are merged at rate
\be
(c_{1}+c_{2}) \theta^{-1/2} \frac{L L'}{\beta |\Lambda|} \equiv r_{\rm m} L L'.
\ee
It is important that the constants $c_{1}$ and $c_{2}$ be the same for all loops and for both the split and merge events. This may seem a daring conjecture to make, but it has been verified numerically in \cite{GLU} in a very similar model. It follows that the lengths of macroscopic loops satisfy an effective split-merge process, and the invariant distribution is Poisson-Dirichlet with parameter $r_{\rm s} / r_{\rm m} = \theta$, see e.g.\ \cite{Bertoin,GUW} and references therein.

The case $u \in (0,1)$ is different because loops split with only half the rate above. Indeed, the appearance of a new transition within the loop may just rearrange it: topologically, this is like $0 \leftrightarrow 8$. This results in PD$(\frac\theta2)$. Notice that this cannot happen when $u=1$, or when $u=0$ on a bipartite graph. These considerations allow to formulate the following conjecture.

\begin{conjecture}
\label{conj joint distr}
Assume that $d, \theta, u, \beta$ are such that macroscopic loops are present. Then, as $|\Lambda|\to\infty$ then $k \to \infty$,
\begin{itemize}
\item the random variable $\sum_{i=1}^{k} \frac{L^{(i)}}{\beta|\Lambda|}$ converges (in probability) to a constant, denoted $\nu$;
\item the sequence of decreasing numbers $\bigl( \frac{L^{(1)}}{\nu \beta |\Lambda|}, \dots, \frac{L^{(k)}}{\nu \beta |\Lambda|} \bigr)$ converges (in probability) to a Poisson-Dirichlet distribution. More precisely, it converges to PD$(\theta)$ if $u=1$ and to PD$(\frac\theta2)$ if $u \in (0,1)$.
\end{itemize}
\end{conjecture}

The case $u=0$ is a bit subtle as it depends on the graph. PD$(\theta)$ is the right choice for bipartite lattices, while PD$(\frac\theta2)$ should be expected otherwise.

It turns out that Conjecture \ref{conj joint distr} is relevant for the discussion about symmetry breaking, even though the heuristics is rather indirect. If $x$ and $y$ are two vertices that are very far apart, the probability that they belong to the same loop is equal to the probability $\nu^{2}$ that they both belong to macroscopic loops, times the probability that they belong to the same element of the corresponding random partition. This can easily be calculated using Beta$(\vartheta)$ i.i.d.\ random variables $X_{1}, X_{2}, \dots$, so that
\[
\Bigl( X_{1}, (1-X_{1}) X_{2}, (1-X_{1}) (1-X_{2}) X_{3}, \dots \Bigr)
\]
has GEM$(\vartheta)$ distribution, which is closely related to PD$(\vartheta)$. Then, when $x$ and $y$ are far apart,
\be
\label{GEM calculation}
\begin{split}
\bbP(E_{x,y,0}) &\approx \nu^{2} \sum_{k\geq1} \bbE \bigl( (1-X_{1})^{2} \dots (1-X_{k-1})^{2} X_{k}^{2} \bigr) \\
&= \nu^{2}  \sum_{k\geq1} \Bigl( \frac\vartheta{\vartheta+2} \Bigr)^{k-1} \frac2{(\vartheta+1) (\vartheta+2)} \\
&= \frac{\nu^{2}}{\vartheta+1}. 
\end{split}
\ee
The approximation should become exact in the limits $|\Lambda|\to\infty$ then $\|x-y\| \to \infty$.

\subsection{Nature of pure Gibbs states}
\label{sec outlook}

Let us focus on the case $u=1$. The Hamiltonian $H_{\Lambda}^{(1)}$ is (minus) the sum of transposition operators. Ferromagnetic product states of the form $\otimes_{x} |a\rangle$, with $a \in \{-S,\dots,S\}$, are ground states: They are eigenstates of each $T_{xy}$ with eigenvalue 1; and this is the largest eigenvalue since $T_{xy}^{2} = \Id$. It is tempting to conclude that a ferromagnetic phase transition takes place (for $d\geq3$) and that the pure Gibbs states are of the form $\langle \cdot \rangle_{\vec\Omega}$ with $\vec\Omega \in \bbS^{2}$:
\be
\langle \cdot \rangle_{\vec\Omega} = \lim_{h\to0+} \lim_{|\Lambda|\to\infty} \langle \cdot \rangle_{H_{\Lambda}^{(1)} + h \sum_{x} \vec\Omega \cdot \vec S_{x}}.
\ee
The pure state $\langle \cdot \rangle_{\vec e_{3}}$ is represented by space-time spin configurations where long loops have spin $S$, while finite loops have any spin values. It follows that
\be
\langle S_{x}^{3} \rangle_{\vec e_{3}} = \nu S, \qquad \langle S_{x}^{1} \rangle_{\vec e_{3}} = \langle S_{x}^{2} \rangle_{\vec e_{3}} = 0.
\ee
Decomposing the rotation-invariant Gibbs state into pure states, and using asymptotic clustering of pure states, we have
\be
\label{heur1}
\begin{split}
\langle S_{x}^{3} S_{y}^{3} \rangle &= \tfrac13 \langle \vec S_{x} \cdot \vec S_{y} \rangle = \tfrac13 \tfrac1{4\pi} \int_{\bbS^{2}} \langle \vec S_{x} \cdot \vec S_{y} \rangle_{\vec\Omega} \dd\vec\Omega \\
&= \tfrac13 \langle \vec S_{x} \cdot \vec S_{y} \rangle_{\vec e_{3}} \approx \tfrac13 \sum_{i=1}^{3} \langle S_{x}^{i} \rangle_{\vec e_{3}} \langle S_{y}^{i} \rangle_{\vec e_{3}} = \tfrac13 \nu^{2} S^{2}.
\end{split}
\ee
On the other hand, using Theorem \ref{thm corr} and Eq.\ \eqref{GEM calculation} with $\vartheta = 2S+1$, we have
\be
\label{heur2}
\langle S_{x}^{3} S_{y}^{3} \rangle = \tfrac13 S(S+1) \bbP(E_{x,y,0}) \approx \tfrac16 \nu^{2} S.
\ee
Eqs \eqref{heur1} and \eqref{heur2} agree in the case $S=\frac12$. This should be expected, as $H_{\Lambda}^{(1)}$ is then the Hamiltonian of the Heisenberg ferromagnet (see Section \ref{sec spin12}). But the equations disagree for all other values of $S$, in particular $S=1$. Eq.\ \eqref{heur2} seems trustworthy as it relies on Conjecture \ref{conj joint distr}. This suggests that the nature of symmetry breaking and the structure of pure Gibbs states are more subtle due to the bigger SU(3) symmetry. Hopefully more light will be shed on these questions in the future.

The case $u=0$ is similar, with the staggered magnetization replacing the magnetization. The calculations above confirm the existence of antiferromagnetic pure states when $S=\frac12$, while the situation for $S \geq 1$ is less clear. In the case $u\in(0,1)$ and $S=\frac12$, we can check that Conjecture \ref{conj joint distr} is compatible with the breaking of the U(1) symmetry: Indeed, let $\bbS^{1}$ denote the unit circle in the plane 1-3; then
\be
\begin{split}
\langle S_{x}^{3} S_{y}^{3} \rangle &= \tfrac12 \langle S_{x}^{1} S_{y}^{1} + S_{x}^{3} S_{y}^{3} \rangle \\
&= \tfrac12 \tfrac1{2\pi} \int_{\bbS^{1}} \langle S_{x}^{1} S_{y}^{1} + S_{x}^{3} S_{y}^{3} \rangle_{\vec\Omega} \dd\vec\Omega = \tfrac12 \langle S_{x}^{1} S_{y}^{1} + S_{x}^{3} S_{y}^{3} \rangle_{\vec e_{3}} \\
&\approx \tfrac12 \langle S_{x}^{1} \rangle_{\vec e_{3}} \langle S_{y}^{1} \rangle_{\vec e_{3}} + \tfrac12 \langle S_{x}^{3} \rangle_{\vec e_{3}} \langle S_{y}^{3} \rangle_{\vec e_{3}} = \tfrac18 \nu^{2}.
\end{split}
\ee
This is compatible with Theorem \ref{thm corr}, $\langle S_{x}^{3} S_{y}^{3} \rangle = \frac14 \bbP(E_{x,y,0})$, and Eq.\ \eqref{GEM calculation} with $\vartheta = 1$.

In the case $S=1$ and $u \in (0,1)$, a similar heuristics should be possible, that would confirm and help characterize the nematic phase that is expected in the quantum model.

\renewcommand{\refname}{\small Handwritten notes}

\newpage
{
\renewcommand{\refname}{\small References}
\bibliographystyle{symposium}

\begin{thebibliography}{99}

\bibitem[N1]{Fro}
J.~Fr\"ohlich,
{\em Phase transitions and continuous symmetry breaking},
Vienna lectures,
available at http://www.maphy.uni-tuebingen.de/$\sim$chha/Froehlich\_ESI\_part1.pdf (2011)

\bibitem[N2]{Toth2}
B.~T\'oth,
{\em Reflection positivity, infrared bounds, continuous symmetry breaking},
Prague lectures,
available at
http://www.math.bme.hu/$\sim$balint/prague\_96/ (1996)

\bibitem[N3]{Uel}
D.~Ueltschi,
{\em Phase transitions in classical and quantum Heisenberg models},
T\"ubingen lectures,
available at http://www.ueltschi.org/articles/12-U.pdf

\end{thebibliography}

\begin{thebibliography}{99}
\vspace{-2em}
\setlength{\columnsep}{2.5em}

\begin{multicols}{2}\scriptsize

\bibitem{AKLT}
I.~Affleck, T.~Kennedy, E.H.~Lieb, H.~Tasaki,
{\em Valence bond ground states in isotropic quantum antiferromagnets},
Comm. Math. Phys. 115, 477-528 (1988)

\bibitem{ALSSY}
M.~Aizenman, E.H.~Lieb, R.~Seiringer, J.P.~Solovej, J.~Yngvason,
{\em Bose-Einstein quantum phase transition in an optical lattice model},
Phys. Rev. A, 70, 023612 (2004)

\bibitem{AN}
M.~Aizenman, B.~Nachtergaele,
{\em Geometric aspects of quantum spin states},
Comm. Math. Phys., 164, 17--63 (1994)

\bibitem{AFFS}
C.~Albert, L.~Ferrari, J.~Fr\"ohlich, B.~Schlein,
{\em Magnetism and the Weiss exchange field --- a theoretical analysis motivated by recent experiments},
J. Statist. Phys. 125, 77--124 (2006)

\bibitem{AK}
G.~Alon, G.~Kozma,
{\em The probability of long cycles in interchange processes},
to appear in Duke Math J.; arXiv:1009.3723 [math.PR]

\bibitem{Ang}
O.~Angel,
{\em Random infinite permutations and the cyclic time random walk},
Discrete Math. Theor. Comput. Sci. Proc., 9--16 (2003)

\bibitem{AZ}
N.~Angelescu, V.A.~Zagrebnov,
{\em A lattice model of liquid crystals with matrix order parameter},
J. Phys. A 15, 639--643 (1982)

\bibitem{BN}
S.~Bachmann, B.~Nachtergaele,
{\em On gapped phases with a continuous symmetry and boundary operators},
arXiv:1307.0716 [math-ph]

\bibitem{BO}
C.D.~Batista, G.~Ortiz,
{\em Algebraic approach to interacting quantum systems},
Adv. Phys. 53, 1--82 (2004)

\bibitem{Ber}
N.~Berestycki,
{\em Emergence of giant cycles and slowdown transition in random transpositions and $k$-cycles},
Electr. J. Probab. 16, 152--173 (2011)

\bibitem{BK}
N.~Berestycki, G.~Kozma,
{\em Cycle structure of the interchange process and representation theory},
to appear in Bull. Soc. Math. France; arXiv:1205.4753 [math.PR]

\bibitem{Bertoin}
J.~Bertoin,
{\em Random fragmentation and coagulation processes},
Cambridge Studies Adv. Math. 102, Cambridge University Press (2006)

\bibitem{BU}
V.~Betz, D.~Ueltschi,
{\em Spatial random permutations and Poisson-Dirichlet law of cycle lengths},
Electr. J. Probab. 16, 1173--1192 (2011)

\bibitem{Bis}
M.~Biskup,
{\em Reflection positivity and phase transitions in lattice spin models},
in Methods of contemporary mathematical statistical physics, Lect. Notes Math. 1970, 1--86 (2009)

\bibitem{BC}
M.~Biskup, L.~Chayes,
{\em Rigorous analysis of discontinuous phase transitions via mean-field bounds},
Commun. Math. Phys. 238, 53--93 (2003)


\bibitem{Bjo}
J.~Bj\"ornberg,
{\em Infrared bounds and mean-field behaviour in the quantum Ising model},
to appear in Commun. Math. Phys.;
arXiv:1205.3385 [math.PR]


\bibitem{BKU}
C.~Borgs, R.~Koteck\'y, D.~Ueltschi,
{\em Low temperature phase diagrams for quantum perturbations of classical spin systems},
Comm. Math. Phys. 181, 409--446 (1996)

\bibitem{CS}
J.~Conlon, J.P.~Solovej,
{\em Upper bound on the free energy of the spin 1/2 Heisenberg ferromagnet},
Lett. Math. Phys. 23, 223--231 (1991)

\bibitem{CI}
N.~Crawford, D.~Ioffe,
{\em Random current representation for transverse field {I}sing model},
Comm. Math. Phys. 296, 447--474 (2010)

\bibitem{DFF}
N.~Datta, R.~Fern\'andez, J.~Fr\"ohlich,
{\em Low temperature phase diagrams of quantum lattice systems. I. Stability for perturbations of classical systems with finitely-many ground states},
J. Stat. Phys. 84, 455--534 (1996)

\bibitem{DLS}
F.J.~ Dyson, E.H.~Lieb, B.~Simon,
{\em Phase transitions in quantum spin systems with isotropic and nonisotropic interactions},
J. Statist. Phys. 18, 335--383 (1978)

\bibitem{FB}
H.~Falk, L.W.~Bruch,
{\em Susceptibility and fluctuation},
Phys. Rev. 180, 442--444 (1969)

\bibitem{FJ}
M.E.~Fisher, D.~Jasnow,
{\em Decay of order in isotropic systems of restricted dimensionality. II. Spin systems},
Phys. Rev. B 3, 907--924 (1971)

\bibitem{FKK}
Yu.A.~Fridman, O.A.~Kosmachev, Ph.N.~Klevets,
{\em Spin nematic and orthogonal nematic states in $S=1$ non-Heisenberg magnet},
J. Magnetism and Magnetic Materials 325, 125--129 (2013)

\bibitem{FILS1}
J.~Fr\"ohlich, R.B.~Israel, E.H.~Lieb, B.~Simon,
{\em Phase transitions and reflection positivity. I. General theory and long range lattice models},
Comm. Math. Phys. 62, 1--34 (1978)

\bibitem{FILS2}
J.~Fr\"ohlich, R.B.~Israel, E.H.~Lieb, B.~Simon,
{\em Phase transitions and reflection positivity. II. Lattice systems with short-range and Coulomb interactions},
J. Statist. Phys. 22, 297--347 (1980)

\bibitem{FP1}
J.~Fr\"ohlich, C.-\'E.~Pfister,
{\em On the absence of spontaneous symmetry breaking and of crystalline ordering in two-dimensional systems},
Comm. Math. Phys. 81, 277--298 (1981)

\bibitem{FP2}
J.~Fr\"ohlich, C.-\'E.~Pfister,
{\em Absence of crystalline ordering in two dimensions},
Comm. Math. Phys. 104, 697--700 (1986)

\bibitem{FSS}
J.~Fr\"ohlich, B.~Simon, T.~Spencer,
{\em Infrared bounds, phase transitions and continuous symmetry breaking},
Comm. Math. Phys. 50, 79--95 (1976)

\bibitem{GKRV}
C.~Giardina, J.~Kurchan, F.~Redig, K.~Vafayi, 
{\em Duality and hidden symmetries in interacting particle systems},
J. Statist. Phys. 135, 25--55 (2009)

\bibitem{Gin}
J.~Ginibre,
{\em Existence of phase transitions for quantum lattice systems},
Comm. Math. Phys. 14, 205-- (1969)

\bibitem{GUW}
C.~Goldschmidt, D.~Ueltschi, P.~Windridge,
{\em Quantum Heisenberg models and their probabilistic representations},
in Entropy and the Quantum II, Contemp. Math. 552, 177--224 (2011); arXiv:1104.0983 [math-ph]

\bibitem{Gri}
G.R.~Grimmett,
{\em Space-time percolation},
in In and out of equilibrium 2, Progr. Probab. 60, 305--320 (2008)

\bibitem{GLU}
S.~Grosskinsky, A.A.~Lovisolo, D.~Ueltschi,
{\em Lattice permutations and Poisson-Dirichlet distribution of cycle lengths},
J. Statist. Phys. 146, 1105--1121 (2012)

\bibitem{GS}
L.-H. Gwa, H.~Spohn, 
{\em Six-vertex model, roughened surfaces, and an asymmetric spin Hamiltonian},
Phys. Rev. Lett. 68, 725--728 (1992)

\bibitem{Ham1}
A.~Hammond,
{\em Infinite cycles in the random stirring model on trees},
to appear in Bull. Inst. Math. Acad. Sinica;
arXiv:1202.1319 [math.PR]

\bibitem{Ham2}
A.~Hammond,
{\em Sharp phase transition in the random stirring model on trees},
arXiv:1202.1322 [math.PR]

\bibitem{Har}
T.E.~Harris,
{\em Nearest-neighbor Markov interaction processes on multidimensional lattices},
Adv. Math. 9, 66--89 (1972)

\bibitem{Iof}
D.~Ioffe,
{\em Stochastic geometry of classical and quantum Ising models},
in Methods of contemporary mathematical statistical physics, Lect. Notes Math. 1970, 87--127 (2009)

\bibitem{ISV}
D.~Ioffe, S.~Shlosman, Y.~Velenik,
{\em 2D Models of statistical physics with continuous symmetry: the case of singular interactions},
Commun. Math. Phys. 226, 433--454 (2002)


\bibitem{JK}
S.~Jansen, N.~Kurt,
{\em On the notion(s) of duality for Markov processes},
arXiv:1210.7193 [math.PR]

\bibitem{Ken}
T.~Kennedy,
{\em Long range order in the anisotropic quantum ferromagnetic Heisenberg model},
Comm. Math. Phys. 100, 447--462 (1985)

\bibitem{KLS1}
T.~Kennedy, E.H.~Lieb, B.S.~Shastry,
{\em Existence of N\'eel order in some spin-$\frac12$ Heisenberg antiferromagnets},
J. Statist. Phys. 53, 1019--1030 (1988)

\bibitem{KLS2}
T.~Kennedy, E.H.~Lieb, B.S.~Shastry,
{\em The XY model has long-range order for all spins and all dimensions greater than one},
Phys. Rev. Lett. 61, 2582--2584 (1988)


\bibitem{KT}
T.~Koma, H.~Tasaki,
{\em Decay of superconducting and magnetic correlations in one-and two-dimensional Hubbard models},
Phys. Rev. Lett. 68, 3248--3251 (1992)

\bibitem{Lef}
R.~Lefevere,
{\em Macroscopic diffusion from a Hamilton-like dynamics},
J. Stat. Phys. 151, 861--869 (2013)

\bibitem{MS}
O.A.~McBryan, T.~Spencer,
{\em On the decay of correlations in $SO(n)$-symmetric ferromagnets},
Comm. Math. Phys. 53, 299--302 (1977)

\bibitem{MW}
N.D.~Mermin, H.~Wagner,
{\em Absence of ferromagnetism or antiferromagnetism in one- or two-dimensional isotropic Heisenberg models},
Phys. Rev. Lett. 17, 1133--1136 (1966)

\bibitem{Mes}
A.~Messiah,
{\em Quantum Mechanics},
Dover (1999)


\bibitem{Nac1}
B.~Nachtergaele,
{\em Quasi-state decompositions for quantum spin systems},
in Probability Theory and Mathematical Statistics, B Grigelionis et al (eds),
pp 565--590 (1994);
arXiv:cond-mat/9312012

\bibitem{Nac2}
B.~Nachtergaele,
{\em A stochastic geometric approach to quantum spin systems},
in Probability and Phase Transitions, G. Grimmett (ed.), Nato Science series C 420, pp 237--246 (1994)

\bibitem{Nac}
B.~Nachtergaele,
lecture notes, unpublished (2012)

\bibitem{NP}
E.J.~Neves, J.F.~Perez,
{\em Long range order in the ground state of two-dimensional antiferromagnets},
Phys. Lett. A 114, 331--333 (1986)

\bibitem{Pfi}
C.-\'E.~Pfister,
{\em On the symmetry of the Gibbs states in two-dimensional lattice systems},
Comm. Math. Phys. 79, 181--188 (1981)

\bibitem{Sch}
O.~Schramm,
{\em Compositions of random transpositions},
Israel J. Math. 147, 221--243 (2005)

\bibitem{SS}
G.~Sch\"utz, S.~Sandow,
{\em Non-Abelian symmetries of stochastic processes: Derivation of correlation functions for random-vertex models and disordered interacting particle systems},
Phys. Rev. E 49, 2726--2741 (1994)

\bibitem{Toth1}
B.~T\'oth,
{\em Improved lower bound on the thermodynamic pressure of the spin $1/2$ Heisenberg ferromagnet},
Lett. Math. Phys. 28, 75--84 (1993)

\bibitem{TLMP}
T.A.~T\'oth, A.M.~L\"auchli, F.~Mila, K.~Penc,
{\em Competition between two- and three-sublattice ordering for $S=1$ spins on the square lattice},
Phys. Rev. B 85, 140403 (2012)

\bibitem{TZX}
H.-H.~Tu, G.-M.~Zhang, T.~Xiang,
{\em Class of exactly solvable SO(n) symmetric spin chains with matrix product ground states},
Phys. Rev. B 78, 094404 (2008)

\end{multicols}\end{thebibliography}

}

\end{document}